\documentclass[11pt,letterpaper]{article}

\usepackage{fullpage}
\usepackage{times}
\usepackage{multirow}
\usepackage{amsmath}
\usepackage{amssymb}
\usepackage{amsthm}
\usepackage{mathrsfs}
\usepackage{array}
\usepackage{subfigure}

\usepackage{tikz}
\usetikzlibrary{arrows,decorations,decorations.shapes,backgrounds,shapes}

\newtheorem{theorem}{Theorem}[section]
\newtheorem{lemma}[theorem]{Lemma}
\newtheorem{corollary}[theorem]{Corollary}
\newtheorem{conjecture}[theorem]{Conjecture}

\theoremstyle{definition}

\newtheorem{mechanism}{Mechanism}

\newcommand{\hx}{\hat{x}}
\newcommand{\hy}{\hat{y}}
\newcommand{\crit}{\text{crit}}
\newcommand{\arc}[1]{(#1)}
\newcommand{\carc}[1]{[#1]}
\newcommand{\pair}[1]{\langle #1 \rangle}
\newcommand{\bigO}{\mathcal{O}}
\newcommand{\bx}{\mathbf{x}}
\renewcommand{\exp}{\mathbb{E}}
\newcommand{\mc}{\text{mc}}
\renewcommand{\sc}{\text{sc}}
\newcommand{\opt}{\text{OPT}}
\renewcommand{\(}{\left(}
\renewcommand{\)}{\right)}
\newcommand{\cost}{\text{cost}}
\newcommand{\lrm}{\text{lrm}}
\newcommand{\rc}{\text{rc}}
\newcommand{\cen}{\text{cen}}

\newcommand{\h}[1]{\hat{#1}}%
\newcommand{\xopt}{x^*_i}
\newcommand{\xnewopt}{\bx^*}
\newcommand{\half}{\frac{1}{2}}
\newcommand{\halfx}[1]{\frac{#1}{2}}
\newcommand{\quartx}[1]{\frac{#1}{4}}

\begin{document}

\title{Strategyproof Approximation Mechanisms\\ for Location on Networks}
\author{
Noga Alon\thanks{
	Microsoft Israel R\&D Center,
	13 Shenkar Street, Herzeliya 46725, Israel,
	and Schools of Mathematics and Computer Science,  
	Tel Aviv University, Tel Aviv, 69978, Israel,
	Email: \texttt{nogaa@tau.ac.il}.
	Research supported in part by a USA Israeli BSF grant, by a grant from
the Israel Science Foundation, by an ERC advanced grant and by the Hermann Minkowski Minerva Center for Geometry at Tel Aviv University.
}
\and
Michal Feldman\thanks{
	Microsoft Israel R\&D Center,
	13 Shenkar Street, Herzeliya 46725, Israel,
	and School of Business Administration and Center for the Study of Rationality, 
	The Hebrew University of Jerusalem, Jerusalem 91904, Israel.
	Email: \texttt{mfeldman@huji.ac.il} 
}
\and
Ariel D. Procaccia\thanks{
	Microsoft Israel R\&D Center,
	13 Shenkar Street, Herzeliya 46725, Israel. 
	Email: \texttt{arielpro@cs.huji.ac.il}.} 
\and 
Moshe Tennenholtz \thanks{
	Microsoft Israel R\&D Center,
	13 Shenkar Street, Herzeliya 46725, Israel,
	and Technion, IIT, 
	Haifa 32000, Israel.  
	Email: \texttt{moshet@microsoft.com}}
}
\date{}
\maketitle

\begin{abstract}
We consider the problem of locating a facility on a network, represented by a graph. A set of strategic agents have different ideal locations for the facility; the cost of an agent is the distance between its ideal location and the facility. A \emph{mechanism} maps the locations reported by the agents to the location of the facility. Specifically, we are interested in social choice mechanisms that do not utilize payments. We wish to design mechanisms that are \emph{strategyproof}, in the sense that agents can never benefit by lying, or, even better, \emph{group strategyproof}, in the sense that a coalition of agents cannot all benefit by lying. At the same time, our mechanisms must provide a small approximation ratio with respect to one of two optimization targets: the social cost or the maximum cost.

We give an almost complete characterization of the feasible truthful approximation ratio under both target functions, deterministic and randomized mechanisms, and with respect to different network topologies. Our main results are: We show that a simple randomized mechanism is group strategyproof and gives a $(2-2/n)$-approximation for the social cost, where $n$ is the number of agents, when the network is a circle (known as a \emph{ring} in the case of computer networks); we design a novel ``hybrid'' strategyproof randomized mechanism that provides a tight approximation ratio of 3/2 for the maximum cost when the network is a circle; and we show that no randomized SP mechanism can provide an approximation ratio better than $2-o(1)$ to the maximum cost even when the network is a tree, thereby matching a trivial upper bound of two.
\end{abstract}

\setcounter{page}{0} \thispagestyle{empty} 
\newpage

\section{Introduction}

We consider a setting which consists of a network (represented by a graph) and a set of strategic agents. We would like to locate a facility on the network, in a way that is based on the preferences of the agents. However, the agents may disagree on the ideal location for the facility; in particular, the \emph{cost} of an agent given a facility location is the distance of the facility from the agent's ideal location. 

This abstract setting has many natural interpretations. For example, a typical social choice scenario concerns a network of roads in a town. The mayor would like to choose a location for a public facility, such as a library or a fire station. Some citizens would like the facility to be close to their homes whereas others may wish it to be adjacent to their workplace, but they all have (presumably different) ideal locations. 

A second, prominent example is obtained by supposing that the network is a telecommunications network, such as a local computer network or the Internet. In this case, the agents are the network users or service providers, and the facility can be, e.g., a filesharing server or a router. This interpretation also motivates attention to specific, common network topologies such as tree networks (also known as hierarchical networks), star networks (which are, graph-theoretically speaking, a special case of trees), and ring networks.

Finally, Schummer and Vohra~\cite{SV04} point out that the network can also be \emph{virtual} rather than physical. For instance, a group of people may wish to schedule a daily meeting or task (such as a file backup). In this case, the network is simply a single cycle around the face of a (24 hour) clock; a person whose ideal time is 23:59 should be indifferent between 23:57 and 00:01.  

A \emph{mechanism} in the above setting is a function that receives the ideal locations of the agents as input, and returns the facility location. Our game-theoretic goal is to design mechanisms that are truthful or \emph{strategyproof} (SP), in the sense that an agent cannot benefit by misreporting its ideal location, regardless of the reports of the other agents. Moreover, we would like our mechanisms to be approximately optimal with respect to a target function, where approximation is defined in the usual sense by looking at the worst-case ratio between the cost of the mechanism's solution and the cost of the optimal solution. Specifically, we are interested in two target functions: the social cost (the sum of agents' distances to the facility) and the maximum cost (the largest distance to the facility). In the road network example minimizing the social cost makes more sense if the facility is a library, whereas one would wish to minimize the maximum cost if the facility were a fire station.

Crucially, if one allows the mechanism to make payments, that is, if money is available, then optimal truthful mechanisms can be obtained by using generic mechanism design techniques, such as augmenting the optimal solution with Vickrey-Clarke-Groves (VCG) payments (see, e.g.,~\cite{Nis07}). However, in many settings money cannot be used to induce truthfulness, due to ethical or legal considerations (see, e.g.,~\cite{SV07}). Moreover, in Internet (or, in general, telecommunications) environments making payments is often technically infeasible due to security and banking issues. Finally, mechanisms like VCG are known to suffer from a host of disadvantages and paradoxes~\cite{Roth07}. Therefore, we will be interested in moneyless mechanisms whose very definition precludes payments. Such mechanisms commonly appear in the social choice literature, where they are sometimes known as \emph{social choice rules}. It is important to note that we do not use approximation to circumvent computational hardness (as our optimization problems are tractable), but rather approximation is employed to obtain strategyproofness (at the expense of the optimality of the solution) without resorting to payments.

\paragraph{Related work.} 
The agenda of \emph{approximate mechanism design without money}, namely the study of \emph{moneyless} SP approximation mechanisms for optimization problems, was recently explicitly advocated by Procaccia and Tennenholtz~\cite{PT09}, but can be traced back to work on incentive compatible learning by Dekel et al.~\cite{DFP08}. The basic setting of Procaccia and Tennenholtz is a very special case of ours, where the network is, very simply, a line. It is trivial that if the target function is the social cost, there is an optimal (moneyless) SP mechanism. However, if one wishes to minimize the maximum cost, the best deterministic SP mechanism has an approximation ratio of two, whereas the best randomized SP mechanism has a ratio of 3/2. These straightforward results are summarized in the left column of Table~\ref{tab:results}. Procaccia and Tennenholtz focus on two extensions of the basic setting: location of two facilities on a line, and location of one facility on a line when each agent controls multiple points. All their results hold only when the agents are located on a line, and they do not consider (more general) networks. 

The model investigated by Schummer and Vohra~\cite{SV04} is identical to ours; their goal is to identify the deterministic SP mechanisms in this setting. Schummer and Vohra give a characterization of all the SP mechanisms when the underlying network is a tree. Furthermore, the authors demonstrate that if the network contains a cycle, any SP mechanism is almost dictatorial, in the sense that one fixed agent dictates the location of the facility. The exact notion of dictatorship is discussed in the sequel. This result is analogous to the celebrated Gibbard-Satterthwaite impossibility theorem~\cite{Gib73,Sat75}, but holds for the case where the preferences of the agents are restricted by the topology of the network. Note that Schummer and Vohra do not consider optimizations problems and approximation, nor do they study randomized mechanisms.  

A significant body of work (see, e.g.,~\cite{HT81,Lab85,BL86}) deals with a model practically identical to ours, and considers several variations on the following question: how bad can a Condorcet point be in terms of its social cost? A Condorcet point is a location in a network that is preferred by more than half the agents to any other location. The quality of the Condorcet point is measured by bounding the ratio between its social cost and the social cost of the optimal location. In this respect the results are reminiscent of approximation, but the approach is descriptive rather than algorithmic. In addition, this line of work does not directly deal with incentives. 

Taking the algorithmic perspective, our work is related to the literature on approximation algorithms for $k$-median and $k$-center problems (see, e.g.,~\cite{BE96,ARR98}). However, these problems are tractable when only one facility must be located (i.e., $k=1$), and this is the case we deal with. As mentioned above, the complexity of our problem (and, hence, the need for approximation) stems from game-theoretic considerations rather than the combinatorial richness of the problem.

\paragraph{Our results and techniques.}

\setlength{\tabcolsep}{2pt}
\begin{table}
\begin{center}
\small{
\begin{tabular}{|c|c|c|c|c|c|}
\cline{3-6}
\multicolumn{2}{c|}{} & \textbf{Line~\cite{PT09}} & \textbf{Tree} & \textbf{Circle} & \textbf{General}\\
\hline
\multirow{3}{*}{\textbf{SC}} & \textbf{Det} & UB 1 GSP & UB 1 GSP & LB $\Omega(n)$ SP~\cite{SV04} & LB $\Omega(n)$ SP~\cite{SV04}\\
\cline{2-6}
 & \multirow{2}{*}{\textbf{Ran}} & \multirow{2}{*}{UB 1 GSP} & \multirow{2}{*}{UB 1 GSP} & UB $2-\frac{2}{n}$ \small{GSP (Thm~\ref{thm:sc_rand_ub}+\ref{thm:gsp})} & UB $2-\frac{2}{n}$ \small{SP (Thm~\ref{thm:sc_rand_ub})}\\
& & & & LB open & LB open\\
\hline
\multirow{4}{*}{\textbf{MC}} & \multirow{2}{*}{\textbf{Det}} & UB 2 GSP & UB 2 GSP & UB 2 GSP & UB 2 GSP\\
& & LB 2 SP & LB 2 SP & LB 2 SP & LB 2 SP\\
\cline{2-6}
 & \multirow{2}{*}{\textbf{Ran}} & UB 3/2 GSP & UB $2-\frac{2}{n+2}$ \small{SP}& UB 3/2 SP (Thm~\ref{thm:mm_circle_ub}) & UB 2 GSP\\
& & LB 3/2 SP & LB $2-o(1)$ \small{SP (Thm~\ref{thm:mm_rand_lb})} & LB 3/2 SP & LB $2-o(1)$ \small{SP (Thm~\ref{thm:mm_rand_lb})}\\
\hline
\end{tabular}
}
\end{center}
\caption{Summary of our results. The columns correspond to the graph topology, whereas the rows correspond to either social cost (SC) or maximum cost (MC). The rows are further divided into deterministic mechanisms (Det) or randomized mechanisms (Ran). The cells contain the bounds, where UB and LB stand for upper bound and lower bound, respectively, while SP and GSP stand for strategyproof and group strategyproof, respectively (ideally the upper bound is GSP while the lower bound is SP). The number of agents is $n$. Cells (or columns) with a citation contain results that either appear in, or follow directly from, previous papers. Cells with a theorem number in parentheses contain our main results.}
\label{tab:results}
\end{table}

This paper is the first to study SP (moneyless) approximation mechanisms for facility location in networks. More generally, it is also a very early, and in our opinion the most technically significant, instance of approximate mechanism design without money. Our almost complete characterization of the feasible SP approximation bounds has given rise to three major results; their proofs are elementary but rather complicated and involve
some subtle combinatorial and probabilistic ideas.

Our results are summarized in Table~\ref{tab:results}. We seek SP mechanisms, but in some cases succeed in obtaining even stronger guarantees by showing \emph{group strategyproofness (GSP)}; a mechanism is GSP if whenever a coalition of agents lies, at least one of the members of the coalition does not gain from the deviation. 

The upper part of the table contains our results regarding minimization of the social cost, given in Section~\ref{sec:sc}. For the case where the network is a line, there is a straightforward optimal deterministic GSP mechanism. In Section~\ref{subsec:sc_tree}, we observe that a generalization of this mechanism is GSP and optimal when the network is a tree. 

On the other hand, when the network is a circle (e.g., a ring network or the face of a clock), it follows directly from the results of Schummer and Vohra~\cite{SV04} that no deterministic SP Mechanism can obtain an approximation ratio better than $\Omega(n)$. However, randomization turns out to be quite powerful in this context. Mechanism~\ref{mech:sc_rand_ub} simply chooses an agent at random and returns its ideal location. We observe that this trivial mechanism has an approximation ratio of $2-2/n$, where $n$ is the number of agents, for the social welfare with respect to any network topology. Our main result in Section~\ref{sec:sc} is that Mechanism~\ref{mech:sc_rand_ub} is GSP when the network is a circle. We show that this result can be explicitly stated as follows (with $S^1$ denoting the unit circle); it may be of interest independently of its game-theoretic interpretation. 
\medskip

\noindent\textbf{Theorem~\ref{thm:gsp}}. Let $x_1,\ldots,x_n,y_1,\ldots,y_n\in S^1$, and denote the distance on $S^1$ between $x\in S^1$ and $y\in S^1$ by $d(x,y)$. Then there exists $i\in \{1,\ldots,n\}$ such that
$$
\sum_{k=1}^n d(x_i,x_k) \leq \sum_{k=1}^n d(x_i,y_k) ~~ .
$$
\medskip

The combinatorial heart of the theorem's proof is the identification of \emph{nearly-antipodal pairs}: pairs of points $\langle x_i,x_j\rangle$ such that there are no other points $x_k$ in the arc between $x_i$ and the antipodal point of $x_j$, and vice versa. Our proof establishes that there exists such a pair that do not both gain from the deviation, that is, the above equation holds with respect to one of its two members. 

In Section~\ref{sec:mc}, we investigate SP mechanisms for minimization of the maximum cost. We observe that even a dictatorship gives us a deterministic GSP 2-approximation mechanism. Moreover, this bound is tight with respect to deterministic mechanisms, even if the network is a line.

In Section~\ref{subsec:mc_circle}, we deal with randomized mechanisms for the case where the network is a circle, and present Mechanism~\ref{mech:mm_circle_ub}. This mechanism works by combining two mechanisms: one is applied when all the agents are located on one semicircle, whereas the other is applied when the agents are not located on one semicircle. We have the following theorem. 
\medskip

\noindent\textbf{Theorem~\ref{thm:mm_circle_ub}.}
Assume that the network is a circle. Then Mechanism~\ref{mech:mm_circle_ub} is an SP 3/2-approximation mechanism for the maximum cost.
\medskip

The result matches an SP lower bound of 3/2. Although the theorem's proof is quite lengthy, we suggest that the truly striking aspect of this result is the mechanism itself and its strategyproofness. Indeed, the mechanism seems to be a coarse hybridization of two mechanisms, where the combination is required for achieving the desired approximation ratio. Although each of the two mechanisms is SP in its own right, it seems to us quite extraordinary that the combined mechanism is SP as well. 

As it turns out, when we wish to minimize the maximum cost and the network is a tree, randomization cannot significantly help us. Indeed, even though we show that there is a randomized SP mechanism with an approximation ratio of $2-2/(n+2)$, we establish the following lower bound, which is our third and final major result.
\medskip

\noindent\textbf{Theorem~\ref{thm:mm_rand_lb}.} 
Let there be $n$ agents. Then there exists a tree network such that no SP randomized mechanism can have an approximation ratio that is smaller than $2-\bigO\left(\frac{1}{2^{\sqrt{\log n}}}\right)$ for the maximum cost.
\medskip

It is interesting to note that, while with respect to minimization of the social cost there is an optimal SP mechanism on trees but not on circles, when the target function is the maximum cost trees are as hard as general graphs (the optimal ratio is two), whereas on a circle a better SP approximation ratio is feasible.

\paragraph{Open problems.} 

There is still a gap with respect to the optimal ratio achievable by randomized SP mechanisms on a circle, when the target is the social cost. We establish a GSP upper bound of $2-2/n$, and it is possible to show a lower bound of $1+\epsilon$ for a small constant $\epsilon>0$. We conjecture that there is an SP lower bound of $2-o(1)$ on a circle. Moreover, while Mechanism~\ref{mech:sc_rand_ub} is GSP on a circle, it is not GSP in general; we conjecture that there is a lower bound of $\Omega(n)$ with respect to GSP mechanisms in general networks. 

A second, small gap has to do with randomized GSP mechanisms on a circle when the target is the maximum cost. Mechanism~\ref{mech:mm_circle_ub} gives an SP upper bound of 3/2. However, the mechanism is not GSP, a fact that is not immediately apparent. The question of whether it is possible to achieve a randomized GSP upper bound of 3/2 on a circle remains open.

\section{Preliminaries}

We use the model of Schummer and Vohra~\cite{SV04}. Let $N=\{1,\ldots,n\}$ be the set of agents. The network is represented by a graph $G$, formalized as follows. The graph is a closed, connected subset of Euclidean space $G\subset \mathbb{R}^k$. The graph is composed of a finite number of closed curves of finite length, known as the \emph{edges}. The extremities of the curves are known as \emph{vertices}. The agents and the facility may be located anywhere on $G$. 

The reader might feel that the traditional, discrete model of graphs is more appropriate. In addition, other related papers (e.g., \cite{Lab85}) consider a similar continuous model, but allow the agents to be located only on the vertices (whereas the facility can be located anywhere). Crucially, all our results, or slight variations thereof, hold under both these alternative models as well.

The \emph{distance} between two points $x,y\in G$, denoted $d(x,y)$, is the length of the minimum-length path between $x$ and $y$, where a \emph{path} is a minimal connected subset of $G$ that contains $x$ and $y$. The center of the path between $x$ and $y$ is denoted $\cen(x,y)$, that is, it is a point $z$ on the path such that $d(x,z)=d(y,z)$. We will also use this notation to denote the center of an interval in $\mathbb{R}$. 

A \emph{cycle} in $G$ is defined to be the union of two paths such that their intersection is equal to the set of both their endpoints. A graph that does not contain cycles is called a \emph{tree}. 

We shall be especially interested in the graph that is a single cycle; we refer to such a graph as a \emph{circle}. Given that $G$ is a circle, we denote the shorter open arc between $x,y\in G$ by $\arc{x,y}$, and the shorter closed arc between $x$ and $y$ by $\carc{x,y}$.\footnote{If $x$ and $y$ are antipodal these arcs are ambiguously defined, and when this is problematic we specify to which arc we are referring.} For every $x\in G$, we denote by $\hx$ the antipodal point of $x$ on $G$, that is, the diametrically opposite point. For two points $x,y\in G$, and in the context of an arc of length less than half the circumference of $G$, we denote the ``clockwise operator'' by $\succeq$, and its strong version by $\succ$; specifically, $x\succeq y$ means that $x$ is clockwise of $y$ on the circle. We believe that this operator is completely intuitive and requires no formal definition, but in case of need the reader may find such a definition in~\cite{SV04}. 

Each agent $i\in N$ has an (ideal) \emph{location} $x_i\in G$. The collection $\bx=\langle x_1,\ldots,x_n\rangle\in G^n$ is referred to as the \emph{location profile}. 

A \emph{deterministic mechanism} is a function $f:G^n\rightarrow G$ that maps the reported locations of the agents to the location of a \emph{facility}. When the facility is located at $y\in G$, the cost of agent $i$ is simply the distance between $x_i$ and $y$:
$$
\cost(y,x_i) = d(x_i,y) ~~ .
$$

A \emph{randomized mechanism} is a function $f:G^n\rightarrow \Delta(G)$, i.e., it maps location profiles to probability distributions over $G$ (which randomly designate the location of the facility). If $f(\bx)=P$, where $P$ is a probability distribution over $G$, then the cost of agent $i$ is the expected distance from $x_i$, 
$$
\cost(P,x_i) = \exp_{y\sim P} [d(x_i,y)] ~~ .
$$

A mechanism is said to be \emph{strategyproof (SP)} if agents can never benefit by lying; formally, for every $\bx\in G^n$, $i\in N$, and deviation $x_i'\in G$, it holds that $\cost(f(x_i',\bx_{-i}),x_i)\geq \cost (f(\bx),x_i)$, where $\bx_{-i}=\langle x_1,\ldots,x_{i-1},x_{i+1},\ldots,x_n\rangle$ is the vector of locations excluding $x_i$. 

A mechanism is \emph{group strategyproof (GSP)} if for every coalition of deviators, there is an agent that does not benefit from the deviation. Formally speaking, for every $\bx\in G^n$, every $S\subseteq N$ and every $\bx_S'\in G^S$, there exists $i\in S$ such that $\cost(f(x_S',\bx_{-S}),x_i)\geq \cost (f(\bx),x_i)$. A stronger notion of group strategyproofness that is rather common in the computer science literature is obtained by requiring that for every deviation, it cannot be the case that at least one member of the coalition strictly gains while the others do not lose. While our GSP results do not hold under the stronger notion, we note that our slightly weaker notion is very common in the social choice literature, since in settings with no money (and, hence, no sidepayments) an agent has no incentive to deviate unless it strictly benefits.

We are interested in optimizing one of two target functions: the social cost, and the maximum cost. The \emph{social cost} of a facility location $y\in G$ with respect to a preference profile $\bx$ is $\sc(y,\bx) = \sum_{i\in N} \cost(y,x_i)$. The social cost of a distribution $P$ with respect to $\bx$ is $\sc(P,\bx) = \exp_{y\sim P}[\sc(y,\bx)]$. The \emph{maximum cost} of $y$ with respect to $\bx$ is $\mc(y,\bx) = \max_{i\in N} \cost(y,x_i)$. Finally, the maximum cost of a distribution $P$ with respect to $\bx$ is naturally defined as $\mc(P,\bx)=\exp_{y\sim P}[\mc(y,\bx)]$.

\section{Social Cost}
\label{sec:sc}

In this section we are interested in finding a facility location that minimizes the social cost. We shall first investigate the case where the graph $G$ is a tree; we shall then study mechanisms on general graphs, with a special emphasis on the case where the underlying graph is a circle.

\subsection{Mechanisms on Trees}
\label{subsec:sc_tree}

It is well-known that, when $n$ agents are located on a line, the mechanism that returns the location of the median agent as the facility location is GSP. This is true since an agent can change the location of the median only by declaring itself to be on the median's opposite side, thus pushing the median away from its true location (and this argument is easily generalized for coalitions of agents). Moreover, the median is also optimal with respect to the social cost, since any other location with distance $\delta$ from the median is further away by $\delta$ from at least $n/2$ agents, and closer by at most $\delta$ to at most $n/2$ agents. Hence, on a line choosing the median is an optimal GSP mechanism. 

If the graph $G$ is, more generally, a tree, things are not much more complicated. Indeed, consider the following mechanism for finding the median of a tree with respect to the location profile $\bx\in G^n$, which has been suggested in other contexts and may be considered folklore (similar ideas have been published at least as early as 1981~\cite{HT81}). 
We first fix an arbitrary node as the root of the tree. Then, as long as the current location has a subtree that contains more than half of the agents, we smoothly move down this subtree. Finally, when we reach a point where it is not possible to move closer to more than half the agents by continuing downwards, we stop and return the current location. 

The fact that the above mechanism is GSP is straightforward, and follows from similar arguments as the ones given for a median on a line: an agent can only modify the location of the mechanism's outcome by pushing the returned facility location away from its true location. It can also be verified that the mechanism returns a location that is optimal in terms of the social cost. To summarize, under the assumption that the graph $G$ is a tree there is an optimal GSP mechanism for the social cost.

\subsection{Mechanisms on Circles and General Graphs}
\label{subsec:sc_gen}

We next consider the case where the graph $G$ may contain cycles. With respect to deterministic mechanisms, the answers that we seek may be found in the paper of Schummer and Vohra. Indeed, Schummer and Vohra deal with a model that is identical to ours. They show that if $G$ contains a cycle $C$, and $f:G^n\rightarrow G$ is an SP rule that is onto $G$, then there is a \emph{cycle dictator}, that is, there is $i\in N$ such that for all $\bx\in C^n$, $f(\bx)=x_i$. In other words, if all the agents are located on $C$ then the dictator single-handedly determines the facility location. 

We claim that this result directly implies a tight SP lower bound of $n-1$ on any graph $G$ that contains a cycle (including a circle). Indeed, if the mechanism is not onto $G$ then let $y\in G$ such that for all $\bx\in G^n$, $f(\bx)\neq y$. Now, consider the profile where all the agents are located at $y$; since the optimal social cost given this profile is zero, the approximation ratio of $f$ must be infinite. Therefore, we may assume that $f$ is onto.

Now, if $f$ is GSP and onto, let $C\subseteq G$ be a cycle. By the result of Schummer and Vohra, there is a cycle dictator, without loss of generality agent 1. Consider a preference profile $\bx\in C$ where $\bx_1=y\in C$ and $\bx_i=z\in C$ for all $i\in N\setminus\{1\}$, where $y\neq z$. We have that the optimal social cost is $d(y,z)$, by locating the facility at $z$, whereas $f(\bx)=x_1=y$, that is, $\sc(f(\bx),\bx)=(n-1)d(y,z)$. Hence, the approximation ratio of $f$ is at least $n-1$. 

Enter randomization, which, as we shall demonstrate, allows us to design an SP mechanism with a constant approximation ratio. We consider the following trivial mechanism. 

\begin{mechanism}[Random Dictator]
\label{mech:sc_rand_ub}
Given $\bx\in G^n$, return a facility location according to the probability distribution that gives probability $1/n$ to the location $x_i$, for all $i\in N$. 
\end{mechanism}

This mechanism is obviously SP, since by deviating an agent can only lose if its own location is chosen, and does not affect the outcome if another's location is selected. We also remark that Mechanism~\ref{mech:sc_rand_ub} is by no means novel; it was directly employed by Meir~\cite{Meir08} in the context of classification, and a slight variation was studied by Procaccia and Tennenholtz~\cite{PT09} for facility location on a line when each agent controls multiple locations. It turns out the mechanism guarantees a constant approximation ratio for the problem at hand with respect to any graph $G$. 

\begin{theorem}
\label{thm:sc_rand_ub}
Let $G$ be a graph, and let $N=\{1,\ldots,n\}$. Then Mechanism~\ref{mech:sc_rand_ub} is an SP $(2-2/n)$-approximation mechanism for the social cost. 
\end{theorem}

\begin{proof}
Given $\bx\in G^n$, let $y\in G$ be the optimal facility location, and denote $\opt=\sc(y,\bx)$. Denoting Mechanism~\ref{mech:sc_rand_ub} by rd (for ``random dictator''), it follows from the triangle inequality that
\begin{align*}
\sc(\text{rd}(\bx),\bx) & = \sum_{i\in N} \frac{1}{n} \sum_{j\in N} d(x_i,x_j) 
\leq \frac{1}{n}\sum_{i\in N}\sum_{j\in N\setminus\{i\}} [d(x_i,y) + d(y,x_j)]\\
&= \frac{1}{n}\sum_{i\in N}\left[(n-1)d(x_i,y) + (\opt-d(y,x_i))\right]
= \opt + \frac{n-2}{n}\sum_{i\in N} d(x_i,y)\\
& = \(2-\frac{2}{n}\)\opt ~~ . \qedhere    
\end{align*}
\end{proof}

A few comments are in order. First, notice that the exact same proof actually shows that Mechanism~\ref{mech:sc_rand_ub} gives an approximation ratio of $2-2/n$ in any metric space (and, of course, the mechanism is still SP). Despite the generality of this theorem and the simplicity of its proof, we have not been able to find this explicit result in the literature, but we do not preclude the possibility that it exists in some form. 

Second, it is easy to see that the analysis in the proof of Theorem~\ref{thm:sc_rand_ub} is tight. Indeed, consider a location profile where $x_1=y\in G$ and $x_i=z\in G$ for all $j\in N\setminus\{1\}$, where $y\neq z$. Then the social cost of the optimal location is $d(y,z)$, whereas the social cost of Mechanism~\ref{mech:sc_rand_ub} is
$$
\frac{1}{n} (n-1)d(y,z) + \frac{n-1}{n} d(z,y) = \(2-\frac{2}{n}\) d(y,z) ~~ .
$$

Third, in general Mechanism~\ref{mech:sc_rand_ub} is not GSP. To see this, let $N=\{1,2,3\}$ and consider a star with three arms, that is, $G=(V,E)$ where $V=\{v,u_1,u_2,u_3)$, $E$ contains $(v,u_i)$ for $i=1,2,3$, and $d(v,u_i)=1$ for $i=1,2,3$. In particular, $d(u_i,u_j)=2$ for $i\neq j$. Let $\bx\in G^n$ such that $x_i=u_i$ for $i\in N$. Then for all $i\in N$, $\cost(\text{rd}(\bx),x_i)=4/3$, where once again we let $\text{rd}(\bx)$ denote the outcome of Mechanism~\ref{mech:sc_rand_ub} given $\bx$. Now, consider the profile $\bx'$ where $x_i'=v$ for all $i\in N$, i.e., the three agents deviate to the center of the star. Since $\text{rd}(\bx')=v$ with probability one, we have that $\cost(\text{rd}(\bx'),x_i)=1$, hence all the agents strictly gain from the joint deviation. 

On the other hand, if $G$ is a line then Mechanism~\ref{mech:sc_rand_ub} is GSP. Indeed, it is quite straightforward that for any deviation, the leftmost member of the deviating coalition and the rightmost member cannot both benefit from the deviation: if the expected distance from one decreases, then the expected distance from the other increases. Much more interestingly, we have the following result. 

\begin{theorem}
\label{thm:gspprime}
Let $G$ be a circle. Then Mechanism~\ref{mech:sc_rand_ub} is GSP. 
\end{theorem}

When analyzing the group strategyproofness of the mechanism, we can assume without loss of generality that the deviating coalition contains all the agents. Indeed, the expected cost of an agent given that a nondeviating agent is selected by the mechanism, and the probability that a nondeviating agent is selected by the mechanism, are both independent of the reports of the deviating agents. Hence (after scaling by a factor of $n$), we can give the following combinatorial, more explicit but equivalent formulation of Theorem~\ref{thm:gspprime}, which may be of independent interest. 

\begin{theorem}
\label{thm:gsp}
Let $G$ be a circle, and let $x_1,\ldots,x_n,y_1,\ldots,y_n\in G$. Then there exists $i\in N$ such that
\begin{equation}
\label{eq:goal}
\sum_{k\in N} d(x_i,x_k) \leq \sum_{k\in N} d(x_i,y_k) ~~ .
\end{equation}
\end{theorem}

The intuition behind the proof of the theorem is as follows. When $G$ is a line, we saw that we can identify a specific pair of agents (the leftmost and rightmost) that cannot both benefit from the deviation; this is also true when the agents are located on one semicircle, which can then be treated as an interval. However, when the agents are not on one semicircle, there are no agents that are easily identified as ``extreme''. The main idea is to recognize pairs of ``almost extreme'' agents which we call \emph{nearly-antipodal}. The proof establishes that there exists a pair of nearly-antipodal agents that do not both benefit from the deviation. The full proof of the theorem is relegated to Appendix~\ref{app:gsp}.

Summarizing and slightly generalizing the results given above, we have that if $G$ is a line or a circle, Mechanism~\ref{mech:sc_rand_ub} is GSP, whereas if there is a vertex of degree at least three it is not (since any such vertex locally looks like the counterexample given above). More precisely, we have the following corollary. 

\begin{corollary}
Let $G$ be a (connected) graph. Then Mechanism~\ref{mech:sc_rand_ub} is GSP if and only if the maximum degree in $G$ is two.
\end{corollary}

For a constant number of agents $n$ on a circle, we were able to design a nontrivial SP mechanism with an approximation ratio slightly better than that of Mechanism~\ref{mech:sc_rand_ub}. However, in general it is natural to ask whether Mechanism~\ref{mech:sc_rand_ub} is asymptotically optimal among SP approximation mechanisms for the social cost, that is: Is there an SP Mechanism whose approximation ratio for the social cost is bounded away from two? We are able to prove a lower bound of $1+\epsilon$ for a small constant $\epsilon>0$ (which we omit due to lack of interest), but the gap is still significant. We conjecture that the answer to our question is negative, even if $G$ is a circle. 

\begin{conjecture}
Let $N=\{1,\ldots,n\}$. Then for any circle $G$ there is no randomized SP $(2-\Omega(1))$-approximation mechanism for the social cost. 
\end{conjecture}

Furthermore, we conjecture that there are graphs for which it is not possible to design a GSP (randomized) mechanism with a sublinear approximation ratio. 

\begin{conjecture}
Let $N=\{1,\ldots,n\}$. Then there exists a graph $G$ such that there is no randomized GSP $o(n)$-approximation mechanism for the social cost. 
\end{conjecture}

\section{Maximum Cost}
\label{sec:mc}

In this section we shall be interested in SP mechanisms that minimize the maximum cost. The problem of designing an optimal SP mechanism is very simple with respect to deterministic mechanisms. Recall that, by the result of Schummer and Vohra~\cite{SV04}, strategyproofness can only be obtained by a dictatorship. Hence, consider the mechanism given by $f(\bx)=x_1$ for all $\bx\in G^n$, that is, a dictatorship of agent 1. This mechanism is clearly GSP. Crucially, this mechanism does quite well in terms of our new optimization goal: it provides a 2-approximation for the maximum cost. To see this, given $\bx\in G^n$, let $y$ be the optimal facility location. Then for all $i\in N$, 
$$
d(x_1,x_i)\leq d(x_1,y) + d(y,x_i) \leq 2\cdot\max\{d(y,x_1),d(y,x_i)\} \leq 2\cdot\mc(y,\bx) ~~ .
$$

On the other hand, a deterministic SP mechanism cannot achieve an approximation ratio better than two, even if the underlying graph $G$ is a line~\cite{PT09}. Since in a general graph any edge is locally a line, this lower bound applies to any graph. In other words, dictatorship gives a tight SP upper bound. In the sequel we shall therefore restrict our attention to randomized mechanisms.

\subsection{Randomized Mechanisms on a Circle}
\label{subsec:mc_circle}

We presently consider the case where the graph $G$ is a circle. An important remark is that, even using randomization, we cannot hope to achieve an SP approximation ratio better than 3/2. Indeed, Procaccia and Tennenholtz~\cite{PT09} have established that a randomized SP mechanism does not yield an approximation ratio smaller than 3/2 on a line. Furthermore, they have provided a straightforward matching GSP upper bound of 3/2 on a line using the \emph{Left-Right-Middle (LRM) Mechanism}: given $\bx\in G^n$, with probability 1/4 return the leftmost agent $\min_{i\in N} x_i$, with probability 1/4 return the rightmost agent $\max_{i\in N} x_i$, and with probability 1/2 return the midpoint of the interval between them, that is, 
$$
\cen\(\min_{i\in N} x_i,\max_{i\in N} x_i\) = \frac{\min_{i\in N} x_i + \max_{i\in N} x_i}{2} ~~ .
$$
The idea behind the strategyproofness of this mechanism is very simple: an agent can only affect the outcome of the mechanism by deviating to a location $x_i'<\min_{i\in N} x_i$ or $x_i'>\max_{i\in N} x_i$. In this case, the agent pushes the left or right boundaries \emph{away} from its location by $\delta$, but in doing so may push the midpoint towards its own location by $\delta/2$. Since the midpoint is selected with probability exactly twice that of each of the boundaries, the two terms cancel out. 

Of course, when the agents are on a circle, in general it is meaningless to refer to the ``leftmost'' or ``rightmost'' agent. However, any semicircle can naturally be treated as an interval, and then the LRM mechanism can be applied. The mechanism that we propose is in fact a hybrid of two mechanisms: the LRM mechanism when the agents are located on one semicircle, and the \emph{Random Center (RC) Mechanism} (defined below) when the agents are not on one semicircle.   

\begin{mechanism}
\label{mech:mm_circle_ub}
Given $\bx\in G^n$:
\begin{enumerate}
\item If $\bx$ is such that the agents are located on one semicircle, i.e., there exist $y,z\in G$ such that for all $i\in N$, $x_i\in\carc{y,z}$, then we execute the LRM Mechanism on the arc $\carc{y,z}$, treating it as an interval with the boundaries $y<z$. 

\item If $\bx$ is such that the agents are not located on one semicircle, we execute the \emph{Random Center (RC) Mechanism}, defined as follows. 
\begin{enumerate}
\item Randomly select a point $y\in G$. 
\item Let $\hx_i$ and $\hx_j$ be the two antipodal points adjacent to $y$, that is, $\hx_i$ is the first antipodal point encountered when walking clockwise from $y$, and $\hx_j$ is the first antipodal point encountered when walking counterclockwise from $y$.
\item Return $\cen(\hx_i,\hx_j)$. 
\end{enumerate}
\end{enumerate}
\end{mechanism}

An equivalent way of thinking about the RC Mechanism is letting the mechanism choose the center of an interval between two adjacent antipodal points with probability proportional to the length of the interval. 

Some technical comments regarding Mechanism~\ref{mech:mm_circle_ub} are in order. Regarding the first item, there may be many choices of $y$ and $z$ such that $x_i\in\carc{y,z}$ for all $i\in N$, but the LRM Mechanism is indifferent to the choice. In the context of the second item, it holds that $y\in\carc{\hx_i,\hx_j}$ by the assumption that in $\bx$ the agents are not on one semicircle, i.e., $y$ is on the same arc whose center we return. Furthermore, the RC mechanism is ambiguously defined when the random point $y$ is an antipodal point itself, but this happens with probability zero. 

\emph{A priori}, it seems there is no reason to suppose that Mechanism~\ref{mech:mm_circle_ub} is SP, as usually such hybridizations of SP mechanisms are not SP as a whole. Nevertheless, we have the following theorem. 

\begin{theorem}
\label{thm:mm_circle_ub}
Assume that $G$ is a circle. Then Mechanism~\ref{mech:mm_circle_ub} is an SP 3/2-approximation mechanism for the maximum cost.
\end{theorem}

The nontrivial part of the theorem's proof is the strategyproofness of the mechanism. Although the proof is long, it revolves around several basic properties of the RC Mechanism. Very generally speaking, one of the ideas at the core of the proof is that the locations occupied by the agents in $\bx$ are special, but only in the sense that for every $i\in N$ the antipodal point $\hx_i$ is among the antipodal points. For this reason, the cost of an agent (assuming that the circumference of the circle is one) under the RC Mechanism is at most 1/4 (Lemma~\ref{lem:quarter}). A second important idea is that, from the point of view of agent $i$, the RC Mechanism essentially chooses a location uniformly on the circle, except for its behavior on the arc between the two antipodal points adjacent to $x_i$ (Lemma~\ref{lem:interval}). The detailed proof is given in Appendix~\ref{app:mm_circle_ub}. 

It is possible to show that Mechanism~\ref{mech:mm_circle_ub} is not GSP, even when the agents are assumed not to be located on one semicircle before and after the deviation (that is, the RC Mechanism is used in both cases). The counterexample is due to Dror Shemesh.

\subsection{Randomized Mechanisms on Trees}
\label{subsec:mc_tree}

In the following we assume that the graph $G$ is a tree. We first observe that randomization allows us to do slightly better than dictatorship, especially when the number of agents is small. Indeed, given $\bx\in G^n$, the \emph{center of $G$ with respect to $\bx$} is a point 
$$
y\in \text{argmin}_{z\in G} \mc(z,\bx) ~~ .
$$ 
It is easy to verify that when $G$ is a tree the center is unique.\footnote{This would not be true in a discrete graph model, but this issue can still be easily circumvented.} Therefore, we can denote the (unique) center of $G$ with respect to $\bx$ by $\cen(G,\bx)$. 

We consider the following mechanism: given $\bx\in G^N$, the distribution on the returned location gives probability $1/(n+2)$ to $x_i$ for each $i\in N$, and probability $2/(n+2)$ to $\cen(G,\bx)$. The fact that the mechanism is SP follows from the fact that when agent $i$ deviates from $x_i$ to $x_i'$, it holds that 
$$
|d(x_i,\cen(G,\bx)) - d(x_i,\cen(G,\bx'))|\leq \frac{d(x_i,x_i')}{2} ~~ ,
$$
and therefore, denoting the above mechanism by $f$,  
$$
\cost(f(\bx'),x_i) - \cost(f(\bx),x_i) \geq  \frac{1}{n}\cdot d(x_i,x_i') - \frac{2}{n}\cdot \frac{d(x,x_i')}{2} = 0 ~~ .
$$
The approximation ratio of the mechanism satisfies 
$$
\frac{\mc(f(\bx),\bx)}{\opt}\leq \frac{\frac{2}{n+2}\cdot \opt + \frac{n}{n+2} \cdot 2\cdot\opt}{\opt} = 2 - \frac{2}{n+2} ~~ .
$$

Despite this small improvement over dictatorship, we shall demonstrate that we cannot do significantly better. In other words, our final major result asserts that an SP mechanism cannot achieve an approximation ratio that is bounded away from two for the maximum cost, even on trees. 

\begin{theorem}
\label{thm:mm_rand_lb}
Let $N=\{1,\ldots,n\}$. Then there exists a tree $G$ such that no SP randomized mechanism can have an approximation ratio that is smaller than $2-\bigO\left(\frac{1}{2^{\sqrt{\log n}}}\right)$ for the maximum cost.
\end{theorem}

The proof of the theorem is given in Appendix~\ref{app:mm_rand_lb}.

\bibliographystyle{plain}

\appendix

\section{Proof of Theorem~\ref{thm:gsp}}
\label{app:gsp}

Assume without loss of generality that the circumference of $G$ is 1. Let $x_1,\ldots,x_n\in G$, and define a multiset $X$ by $X=\{x_1,\ldots,x_n\}$. We first note that we can assume that there are no $x_i,x_j\in X$ such that $x_j=\hx_i$. Indeed, in this case the claim holds trivially with respect to either $i$ or $j$, since for all $z\in G$,
$$
d(x_i,z) + d(x_j,z) = 1/2 ~~ .
$$
In particular, for every $x_i,x_j\in X$, $\arc{x_i,x_j}$ and $\arc{\hx_i,\hx_j}$ are well-defined.

We say that two points $x_i,x_j\in X$ are \emph{nearly antipodal} if there is no point $x_k\in X$ such that $x_k\in \arc{x_i,\hx_j}$ or $x_k\in \arc{x_j,\hx_i}$; let $A\subseteq X^2$ be the set of all nearly antipodal pairs. Given a nearly antipodal pair $\pair{x_i,x_j}\in A$, let the \emph{critical arc of $\pair{x_i,x_j}$}, denoted $\crit(x_i,x_j)$, be the \emph{long} open arc between $\hx_i$ and $\hx_j$, that is,
$$
\crit(x_i,x_j)=G\setminus\carc{\hx_i,\hx_j}=\arc{x_i,\hx_j}\cup \carc{x_i,x_j}\cup \arc{x_j,\hx_i} ~~ .
$$
See Figure~\ref{fig:gsp} for an illustration of the construction given above.

\begin{figure}[t]
\begin{center}
\begin{tikzpicture}[scale=1.5]

\tikzstyle{blackdot}=[circle,draw=black,fill=black,thin,
inner sep=0pt,minimum size=1.5mm]
\tikzstyle{whitedot}=[circle,draw=black,fill=white,thin,
inner sep=0pt,minimum size=1.5mm]

\draw (0,0) circle (1cm);

\node (x1) at (-1,0) [blackdot] {};
\draw +(-1.3,0) node {\small{$x_1$}};
\node (bx1) at (1,0) [whitedot] {};
\draw +(1.3,0) node {\small{$\hx_1$}};
\draw[dashed] (x1) -- (bx1);

\node (x2) at (150:1cm) [blackdot] {};
\draw (150:1.3cm) node {\small{$x_2$}};
\node (bx2) at (330:1cm) [whitedot] {};
\draw (330:1.3cm) node {\small{$\hx_2$}};
\draw[dashed] (x2) -- (bx2);

\node (x3) at (45:1cm) [blackdot] {};
\draw (45:1.3cm) node {\small{$x_3$}};
\node (bx3) at (225:1cm) [whitedot] {};
\draw (225:1.3cm) node {\small{$\hx_3$}};
\draw[dashed] (x3) -- (bx3);

\node (x4) at (0,-1) [blackdot] {};
\draw (0,-1.3) node {\small{$x_4$}};
\node (bx4) at (0,1) [whitedot] {};
\draw (0,1.3) node {\small{$\hx_4$}};
\draw[dashed] (x4) -- (bx4);
%
%

\end{tikzpicture}
\end{center}
\caption{An illustration of the construction in the proof of Theorem~\ref{thm:gsp}, for $n=4$. The nearly antipodal pairs are $A=\{\langle x_1,x_3\rangle,\langle x_2,x_4\rangle,\langle x_3,x_4\rangle\}$.}
\label{fig:gsp}
\end{figure}
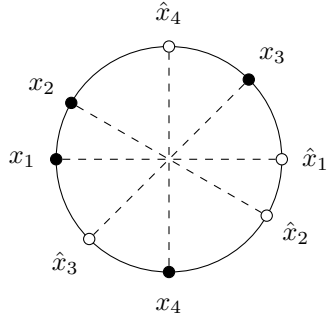

Let $y_1,\ldots,y_n\in G$, and define a multiset $Y$ by $Y=\{y_1,\ldots,y_n\}$. It is sufficient to prove that there exists a pair of nearly antipodal points $\pair{x_i,x_j}\in A$ such that
$$
\sum_{k\in N} d(x_i,x_k)+ \sum_{k\in N} d(x_j,x_k) \leq \sum_{k\in N} d(x_i,y_k) + \sum_{k\in N} d(x_j,y_k) ~~ .
$$
Indeed, in this case we get that Equation~\eqref{eq:goal} holds with respect to either $x_i$ or $x_j$. Therefore, assume for the purpose of contradiction that for every pair of nearly antipodal points $\pair{x_i,x_j}\in A$,
\begin{equation}
\label{eq:cont}
\sum_{k\in N} d(x_i,x_k)+ \sum_{k\in N} d(x_j,x_k) > \sum_{k\in N} d(x_i,y_k) + \sum_{k\in N} d(x_j,y_k) ~~ .
\end{equation}

We claim that Equation~\eqref{eq:cont} implies that for every pair of nearly antipodal points $\pair{x_i,x_j}\in A$, the number of points from $Y$ on $\crit(x_i,x_j)$ is strictly greater than the number of points from $X$ on the same arc. Formally, for $\pair{x_i,x_j}\in A$, let
$$
\alpha^X_{ij} = |\{x_k\in X:\ x_k\in \crit(x_i,x_j)\}| ~~ ,
$$
and
$$
\alpha^Y_{ij} = |\{y_k\in Y:\ y_k\in\crit(x_i,x_j)\}| ~~ .
$$
We have the following claim.

\begin{lemma}
\label{lem:less}
Let $\pair{x_i,x_j}\in A$. Then $\alpha^Y_{ij} > \alpha^X_{ij}$.
\end{lemma}

\begin{proof}
For every point $z\in\carc{x_i,x_j}$, we have that
$$
d(x_i,z) + d(x_j,z) = d(x_i,x_j) ~~ .
$$
Let $d'(x_i,x_j)$ be the length of the longer arc $G\setminus \carc{x_i,x_j}$ between $x_i$ and $x_j$, namely
$$
d'(x_i,x_j) = d(x_i,\hx_j) + d(\hx_j,\hx_i) + d(\hx_i,x_j) = d(x_i,x_j) + 2\cdot d(x_i,\hx_j) > d(x_i,x_j) ~~ .
$$
For every point $z\in \carc{\hx_i,\hx_j}$ it holds that
$$
d(x_i,z) + d(x_j,z) = d'(x_i,x_j) ~~ .
$$
Finally, it holds that for every $z\in\arc{x_i,\hx_j}\cup\arc{x_j,\hx_i}$,
$$
d(x_i,x_j) < d(x_i,z) + d(x_j,z) < d'(x_i,x_j) ~~ .
$$
Since $x_i$ and $x_j$ are nearly antipodal, there are no points from $X$ in $\arc{x_i,\hx_j}$ and $\arc{x_j,\hx_i}$. Therefore,
\begin{equation}
\label{eq:X}
\sum_{k\in N} d(x_i,x_k)+ \sum_{k\in N} d(x_j,x_k) = \alpha^X_{ij} \cdot d(x_i,x_j) + (n-\alpha^X_{ij})\cdot d'(x_i,x_j) ~~ .
\end{equation}
On the other hand,
\begin{equation}
\label{eq:Y}
\sum_{k\in N} d(x_i,y_k) + \sum_{k\in N} d(x_j,y_k) \geq \alpha^Y_{ij}\cdot d(x_i,x_j) + (n-\alpha^Y_{ij})\cdot d'(x_i,x_j) ~~ .
\end{equation}
Using Equations \eqref{eq:X} and \eqref{eq:Y}, we get that \eqref{eq:cont} directly implies that $\alpha^Y_{ij} > \alpha^X_{ij}$, as claimed.
\end{proof}

From Lemma~\ref{lem:less}, we immediately get that
\begin{equation}
\label{eq:less}
\sum_{(x_i,x_j)\in A} \alpha^X_{ij} < \sum_{(x_i,x_j)\in A} \alpha^Y_{ij} ~~ .
\end{equation}

In order to derive a contradiction, we also need the following lemma.

\begin{lemma}
\label{lem:contra}
There exists $r\in\mathbb{N}$ such that
\begin{equation}
\label{eq:rn}
\sum_{(x_i,x_j)\in A} \alpha^X_{ij} = r\cdot n ~~ ,
\end{equation}
whereas
\begin{equation}
\label{rn2}
\sum_{(x_i,x_j)\in A} \alpha^Y_{ij} \leq r\cdot n ~~ ,
\end{equation}
\end{lemma}

\begin{proof}
It is easy to see that $|A|$ is odd (e.g., by induction on $n$); let $|A|=2s+1$, for some $s\in\mathbb{N}$. We first wish to claim that every $x_i\in X$ is a member of exactly $s+1$ critical arcs, which directly proves Equation~\eqref{eq:rn} with $r=s+1$.

Without loss of generality we prove the claim with respect to $x_1\in X$. Consider the clockwise closed arc between $x_1$ and $\hx_1$. Let
$
Z=\{z_1,\ldots,z_t\}
$
be all the points $x_i$ or $\hx_i$ on this arc, where for all $k$, $z_{k+1}\succeq z_k$. In particular, $z_1=x_1$ and $z_t=\hx_1$. For instance, in Figure~\ref{fig:gsp} we have that $Z=\{x_1,x_2,\hx_4,x_3,\hx_1\}$.

Now, we have that the set of nearly-antipodal pairs $A$ is exactly the set of pairs $\pair{x_i,x_j}$ such that $z_k$ is a point $x_i$ and and $z_{k+1}$ is an antipodal point $\hx_j$ (this is a \emph{type 1} nearly-antipodal pair), or $z_k$ is an antipodal point $\hx_i$ and $z_{k+1}$ is a point $x_j$ (this is a \emph{type 2} nearly-antipodal pair). If $\pair{x_i,x_j}$ is a nearly-antipodal pair of type 1, we have that $x_1\in \carc{x_i,x_j}$, and hence $x_1\in \crit(x_i,x_j)$. On the other hand, if $\pair{x_i,x_j}$ is a nearly-antipodal pair of type 2, then $x_1\notin \crit(x_i,x_j)$. Since $z_1=x_1$ is a point from $X$ and $x_{n+1}=\hx_1$ is an antipodal point, the number of nearly-antipodal pairs of type 1 is exactly $s+1$, which proves the claim.

In order to prove Equation~\eqref{rn2}, let $y\in G$. It is sufficient to prove that there exists $x_i\in X$ such that $y$ appears in at most as many critical arcs as $x_i$, since we already know that $x_i$ is a member of exactly $s+1$ critical arcs. We consider the two points or antipodal points that are adjacent to $y$, and briefly examine four cases.
\begin{enumerate}
\item $x_i\preceq y \preceq x_j$: $y$ appears in exactly the critical arcs that contain $x_i$ (these are also exactly the critical arcs that contain $x_j$).
\item $x_i \preceq y \prec \hx_j$: $y$ appears in exactly the critical arcs that contain $x_i$.
\item $\hx_i \prec y \preceq x_j$: $y$ appears in exactly the critical arcs that contain $x_j$.
\item $\hx_i \preceq y \preceq \hx_j$: When walking counterclockwise from $\hx_i$, let $x_k\in X$ be the first point from $X$, and let $\hx_l$ be the last antipodal point such that $x_k\prec \hx_l\preceq \hx_i$. Then $y$ is contained in all the critical arcs that contain $x_k$, except for $\crit(x_k,\hx_l)$, that is, in exactly $s$ critical arcs.
\end{enumerate}
We deduce that every $y_i$ is contained in at most $r=s+1$ critical arcs, which implies the validity of Equation~\eqref{rn2}.
\end{proof}

It follows from Lemma~\ref{lem:contra} that
$$
\sum_{(x_i,x_j)\in A} \alpha^X_{ij} \geq \sum_{(x_i,x_j)\in A} \alpha^Y_{ij} ~~ ,
$$
in contradiction to Equation~\eqref{eq:less}.
\qed

\section{Proof of Theorem~\ref{thm:mm_circle_ub}}
\label{app:mm_circle_ub}

In the proof we assume without loss of generality that the circumference of the circle $G$ is 1. In addition, we denote the outcome of the LRM Mechanism and the RC Mechanism given $\bx\in G^n$ by $\lrm(\bx)$ and $\rc(\bx)$, respectively. 

\begin{lemma}
Mechanism~\ref{mech:mm_circle_ub} is a 3/2-approximation mechanism for the maximum cost.
\end{lemma}

\begin{proof}
Assume first that $\bx\in G^n$ is such that not all agents are located on one semicircle. Let $\alpha$ be the length of the longest arc between two adjacent agents, and assume without loss of generality that these agents are agents 1 and 2. Since the agents are not located on one semicircle, it holds that $\alpha\leq 1/2$. It can be verified that the optimal facility location is $\cen(\hx_1,\hx_2)$, hence we have that $\opt = (1-\alpha)/2$. The mechanism selects the optimal solution with probability $\alpha$, and with probability $1-\alpha$ selects a solution with maximum cost at most 1/2. Therefore, the approximation ratio is at most
\begin{equation}
\label{eq:approx}
\frac{\mc(\rc(\bx),\bx)}{\opt}\leq\frac{\alpha\cdot\frac{1-\alpha}{2} + (1-\alpha)\cdot\frac{1}{2}}{\frac{1-\alpha}{2}} = 1+\alpha \leq \frac{3}{2} ~~ .
\end{equation}

If $\bx$ is such that all the agents are located on one semicircle, then the LRM mechanism is applied: we choose the optimal location with probability 1/2, and a location with twice the optimal cost with probability 1/2, hence the approximation ration is, once again, 3/2.
\end{proof}

In order to establish the strategyproofness of Mechanism~\ref{mech:mm_circle_ub}, we must examine four types of lies: an agent deviating such that before the deviation all the agents were located on one semicircle and after the deviation they are located on one semicircle---``semicircle to semicircle'' (the LRM Mechanism is applied to both); not semicircle to semicircle (RC to LRM); semicircle to not semicircle (LRM to RC); and not semicircle to not semicircle (RC to RC). The semicircle to semicircle case is relatively straightforward, and we tackle it first.  

\begin{lemma}[Semicircle to semicircle]
\label{lem:semi_to_semi}
Assume that $\bx\in G^n$ is such that the agents are on one semicircle, and agent $i$ deviates such that in the new location profile $\bx'$ the agents are also on one semicircle. Then
$$
\cost(\lrm(\bx),x_i)\leq \cost(\lrm(\bx'),x_i) ~~ .
$$
\end{lemma}

\begin{proof}
Let $x_1 \succeq x_2 \succeq \ldots \succeq x_n$, and denote $x_1=l$ (for ``left'') and $x_n=r$ (for ``right''). Suppose agent $i$ deviates from $x_i$ to $x_i'$, such that the agents are on one semicircle. If $\arc{x_i',x_i}$ intersects with the new semicircle, the proof follows directly from the fact that the LRM Mechanism is SP when applied to an interval~\cite{PT09}.

It is easy to verify that $\arc{x_i',x_i}$ may lie in the complement of the new semicircle only if (i) the deviating agent is $r$ and $x_i' \succ l$; or (ii) the deviating agent is $l$ and $x_i' \prec r$. We prove the lemma for the former case, but note that the latter case is completely analogous.

Indeed, suppose without loss of generality that $x_i=r$ and $x_i' \succ l$. Let $r'$ be the adjacent agent to $r$ such that $r' \succ r$. It must hold that $x_i' \preceq \h{r}'$ (otherwise in the new location profile the agents are not all on one semicircle). We denote $\alpha = d(r,r')$, $\beta=d(x_i',l)$, and $\gamma=d(l,r)$ (see Figure~\ref{fig:semi_to_semi}).
From the assumptions of the lemma it follows that $\gamma \leq \half$.
The short arc from $x_i'$ to $x_i$ lies in the complement of the new semicircle if and only if $1 - \beta - \gamma \leq \beta + \gamma$, that is, $\beta+\gamma\geq 1/2$.

\begin{figure}
\centering
\subfigure[Truthful location profile $\bx$.]
{
\label{sub:aa_1}

\begin{tikzpicture}[scale=1.5]

\tikzstyle{blackdot}=[circle,draw=black,fill=black,thin,
inner sep=0pt,minimum size=1.5mm]
\tikzstyle{whitedot}=[circle,draw=black,fill=white,thin,
inner sep=0pt,minimum size=1.5mm]

\draw (0,0) circle (1cm);

\node (hx1) at (1,0) [blackdot] {};
\draw +(1.5,0) node {\small{$x_i=r$}};

\node at (300:1cm) [blackdot] {};
\draw (300:1.3cm) node {\small{$$}};

\node at (270:1cm) [blackdot] {};
\draw (270:1.3cm) node {\small{$$}};

\node at (250:1cm) [blackdot] {};
\draw (250:1.3cm) node {\small{$l$}};

\draw [dashed] (-1,0) -- (1,0);
\node at (-1,0) [whitedot] {};

\end{tikzpicture}
}
\hspace{1cm}
\subfigure[Manipulated location profile $\bx'$.]
{

\label{sub:bb_1}

\begin{tikzpicture}[scale=1.5]

\tikzstyle{blackdot}=[circle,draw=black,fill=black,thin,
inner sep=0pt,minimum size=1.5mm]
\tikzstyle{whitedot}=[circle,draw=black,fill=white,thin,
inner sep=0pt,minimum size=1.5mm]
\tikzstyle{graydot}=[circle,draw=gray,fill=gray,thin,
inner sep=0pt,minimum size=1.5mm]

\draw (0,0) circle (1cm);

\node (hx1) at (1,0) [graydot] {};
\draw +(1.5,0) node {\small{$x_i=r$}};

\node at (300:1cm) [blackdot] {};
\draw (300:1.3cm) node {\small{$r'$}};

\node at (270:1cm) [blackdot] {};
\draw (270:1.3cm) node {\small{$$}};

\node at (250:1cm) [blackdot] {};
\draw (250:1.3cm) node {\small{$l$}};

\node at (150:1cm) [blackdot] {};
\draw (150:1.3cm) node {\small{$x_i'$}};

\draw (330:0.85cm) node {\small{$\alpha$}};
\draw (200:0.85cm) node {\small{$\beta$}};

\draw [dashed] (300:1cm) -- (120:1);
\node at (120:1) [whitedot] {};

\end{tikzpicture}

}
\caption{Illustration of the proof of Lemma~\ref{lem:semi_to_semi}.}
\label{fig:semi_to_semi}
\end{figure}
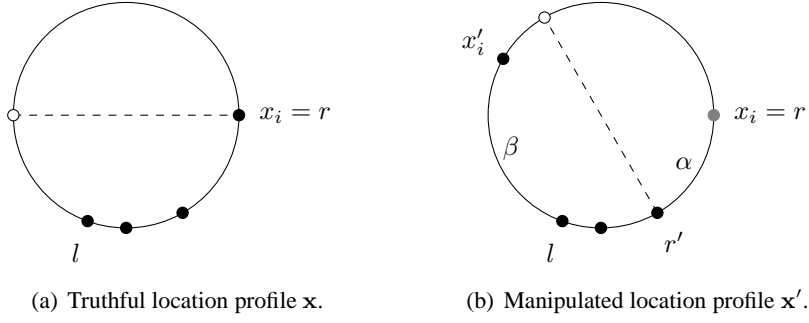

We first calculate the cost of $x_i$ in the location profile $\bx$:
\[
\cost(\lrm(\bx),x_i) = \quartx{1} \cdot \gamma + \quartx{1} \cdot 0 + \half \cdot \halfx{\gamma} = \halfx{\gamma} ~~ .
\]
We distinguish between two cases: 

\noindent\emph{Case 1:} $\alpha \leq 1 - \beta - \gamma$. In this case, the cost of the LRM Mechanism in the new location profile with respect to $x_i$ is
\[
\cost(\lrm(\bx'),x_i) = \quartx{1} \cdot \alpha + \quartx{1} \cdot (1 - \beta - \gamma) + \half \cdot \left( \alpha + \halfx{\gamma - \alpha + \beta}\right) = \quartx{1} + \halfx{\alpha} ~~ . 
\]
It holds that 
$$
\cost(\lrm(\bx),x_i) \leq  \cost(\lrm(\bx'),x_i)\Leftrightarrow\gamma \leq \half + \alpha ~~ ,
$$
but this holds since $\gamma \leq \half$.

\noindent\emph{Case 2:} $\alpha > 1 - \beta - \gamma$. In this case, the cost of the left-middle-right mechanism in the new location profile with respect to $x_i$ is
\[
\cost(\lrm(\bx'),x_i) = \quartx{1} \cdot \alpha + \quartx{1} \cdot (1 - \beta - \gamma) + \half \cdot \left( 1 - \beta - \gamma + \halfx{\gamma - \alpha + \beta}\right) = \quartx{3} - \halfx{\beta} - \halfx{\gamma} ~~ .
\]
It holds that 
$$
\cost(\lrm(\bx),x_i) \leq  \cost(\lrm(\bx'),x_i) \Leftrightarrow \halfx{\gamma} \leq \quartx{3} - \halfx{\beta} - \halfx{\gamma} \Leftrightarrow \gamma \leq \quartx{3} - \halfx{\beta} ~~ .
$$ 
Since $\gamma \leq \half$, it is sufficient to show that $\half \leq \quartx{3} - \halfx{\beta}$, which holds if and only if $\beta \leq \half$. The last inequality follows from the fact that all the agents are on one semicircle after the deviation; this concludes the proof of the lemma.
\end{proof} 

The other three deviations are more complicated, and their proofs require laying some foundations first. We start with two simple lemmata concerning partitions of intervals on the line, both of which will prove useful in several points in the sequel. On a circle, the arc between two points $x,y\in G$ such that $d(x,y)\leq 1/2$ can be regarded as an interval. 

\begin{lemma}
\label{lem:partition}
Let $x,y\in[0,1]\subset \mathbb{R}$. It holds that
$$
d(x,1/2)\leq y\cdot d(x,\cen(0,y)) + (1-y) \cdot d(x,\cen(y,1)) ~~ .
$$
\end{lemma}

\begin{proof}
Assume without loss of generality that $x\leq 1/2$. We distinguish between two cases.

\noindent\emph{Case 1:} $y/2\geq x$. In this case,
$$
y\cdot d(x,\cen(0,y)) + (1-y) \cdot d(x,\cen(y,1)) = y\cdot \(\frac{y}{2}-x\) + (1-y)\cdot \(y+\frac{1-y}{2}-x\) = \frac{1}{2}-x ~~ .
$$
\noindent\emph{Case 2:} $y/2< x$. We have that
$$
y\cdot d(x,\cen(0,y)) + (1-y) \cdot d(x,\cen(y,1)) = y\cdot \(x-\frac{y}{2}\) + (1-y)\cdot \(y+\frac{1-y}{2}-x\) = \frac{1}{2}-x + (2xy-y^2) ~~ .
$$
Since $y/2< x$, it holds that $2xy> y^2$, hence $1/2-x + (2xy-y^2)> 1/2-x$.
\end{proof}

\begin{lemma}
\label{lem:interval}
Let $y_1,\ldots,y_{m+1}\in[0,1/2]$ such that $y_1=0$, $y_{m+1}=1/2$. For all $i=1,\ldots,m$, define $d_i=y_{i+1}-y_i$. Then $$
\sum_{i=1}^m \(d_i\cdot\(\sum_{j=1}^{i-1} d_j + \frac{d_i}{2}\)\)= \frac{1}{8} ~~ .
$$
\end{lemma}

The intuition behind Lemma~\ref{lem:interval} is that choosing the center of each interval with probability equal to the length of the interval is like randomly choosing a point in $[0,1/2]$ with probability $1/2$. The expected distance of a random point in $[0,1/2]$ from 0 is $1/4$, multiplying by $1/2$ yields $1/8$. 

\begin{proof}[Proof of Lemma~\ref{lem:interval}]
By reorganizing the terms, it can be verified that
$$
\sum_{i=1}^m \(d_i\cdot\(\sum_{j=1}^{i-1} d_j + \frac{d_i}{2}\)\) = \frac{\(\sum_{i=1}^m d_i\)^2}{2} = \frac{1}{8} ~~ .
$$
\end{proof}

The next lemma implies that under the Random Center Mechanism, the cost of an agent is at most $1/4$. It follows from the two previous lemmata.

\begin{lemma}
\label{lem:quarter}
For all $\bx\in G^n$ such that $\bx$ is not on one semicircle, and for all $i\in N$,
$$
\cost(\rc(\bx),x_i)\leq \frac{1}{4} ~~ .
$$
\end{lemma}

\begin{proof}
Assume that the locations of the agents satisfy $\hx_1\succeq \hx_2\cdots \succeq \hx_n \succeq \hx_1$, and that there are no antipodal points in $\arc{\hx_i,\hx_{i+1}}$ for all $i\in N$. Let $i\in N$, and assume without loss of generality that $x_i\in\carc{\hx_1,\hx_2}$. In addition, let $\alpha=d(\hx_1,x_i)$ and $\beta=d(x_i,\hx_2)$ (see Figure~\ref{fig:quarter}).

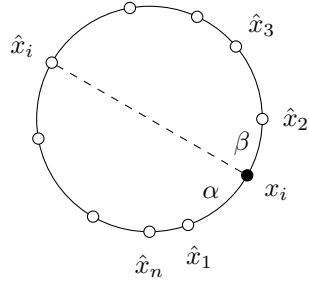
\begin{figure}[t]
\begin{center}
\begin{tikzpicture}[scale=1.5]

\tikzstyle{blackdot}=[circle,draw=black,fill=black,thin,
inner sep=0pt,minimum size=1.5mm]
\tikzstyle{whitedot}=[circle,draw=black,fill=white,thin,
inner sep=0pt,minimum size=1.5mm]

\draw (0,0) circle (1cm);

\node (hx1) at (1,0) [whitedot] {};
\draw +(1.3,0) node {\small{$\hx_2$}};

\node (x2) at (330:1cm) [blackdot] {};
\draw (330:1.3cm) node {\small{$x_i$}};
\node (hx2) at (150:1cm) [whitedot] {};
\draw (150:1.3cm) node {\small{$\hx_i$}};
\draw[dashed] (x2) -- (hx2);

\node at (290:1cm) [whitedot] {};
\draw (290:1.3cm) node {\small{$\hx_1$}};

\draw (310:0.85cm) node {\small{$\alpha$}};
\draw (345:0.85cm) node {\small{$\beta$}};

\node at (270:1cm) [whitedot] {};
\draw (270:1.3cm) node {\small{$\hx_n$}};

\node at (40:1cm) [whitedot] {};
\draw (40:1.3cm) node {\small{$\hx_3$}};

\node at (100:1cm) [whitedot] {};
\node at (65:1cm) [whitedot] {};

\node at (190:1cm) [whitedot] {};
\node at (240:1cm) [whitedot] {};

\end{tikzpicture}
\end{center}
\label{fig:quarter}
\caption{Illustration of the proof of Lemma~\ref{lem:quarter}.}
\end{figure}

We wish to calculate the cost of the Random Center Mechanism with respect to agent $i$. We can break this cost down into two components: the cost when choosing a point in the arc $\carc{\hx_1,\hx_2}$ (this happens with probability $\alpha+\beta$), and the rest of the cost. Notice that $x_i\in\carc{\hx_1,\hx_2}$, i.e., the length of the arc between $\hx_1$ and $\hx_2$ that includes $x_i$ is is at most 1/2; this holds by the assumption that the agents are not on one semicircle. Therefore we can treat $\carc{\hx_1,\hx_2}$ as an interval. By normalizing $d(\hx_1,\hx_2)$, we get from Lemma~\ref{lem:partition} that if we partition the arc $\carc{\hx_1,\hx_2}$ into two arcs of length $y$ and $1-y$ and choose the center of each with probability $y$ and $1-y$, respectively, the cost of agent $i$ can only increase. In particular, this is true when $y=\alpha$. Hence, we have that
\begin{equation}
\label{eq:antipod}
\cost(\rc(\bx),x_i)\leq \cost(\rc(\bx\cup\{\hx_i\}),x_i) ~~ ,
\end{equation}
where $\bx\cup\{\hx_i\}$ is the profile $\bx$ with an additional agent at $\hx_i$, that is, an additional antipodal point at $x_i$. It is sufficient to show that 
$$
\cost(\rc(\bx\cup\{\hx_i\}),x_i) \leq \frac{1}{4} ~~ .
$$
The expression $\cost(\rc(\bx\cup\{\hx_i\}),x_i)$ is the expected distance from agent $i$ when the center of one of the arcs
\begin{equation}
\label{eq:arcs}
\carc{x_i,\hx_2},\carc{\hx_2,\hx_3},\carc{\hx_3,\hx_4},\ldots,\carc{\hx_n,\hx_1},\carc{\hx_1,x_i} ~~,
\end{equation}
is chosen, where the probability of choosing the center of an arc is its length. In order to make this explicit, denote $d_i = d(\hx_i,\hx_{i+1})$ for $i=1,\ldots,n-1$, $d_n=d(\hx_n,\hx_1)$. We have
\begin{equation}
\label{eq:uniform2}
\cost(\rc(\bx\cup\{\hx_i\}),x_i) = \frac{\alpha^2}{2} + \frac{\beta^2}{2} + \sum_{j=2}^{i-1} \(d_j \cdot\(\beta + \sum_{k=2}^{j-1} d_k + \frac{d_j}{2}  \)  \) + \sum_{j=i}^{n} \(d_j \cdot\(\alpha + \sum_{k=j+1}^{n} d_k + \frac{d_j}{2}  \)  \)
\end{equation}

We partition the expression in the right hand side of Equation~\eqref{eq:uniform2} into two sums, each corresponding to the cost of the mechanism on a semicircle of length 1/2, and apply Lemma~\ref{lem:interval} to each, with $y_1=x_i$ in both cases, and $d_1=\alpha$ or $d_1=\beta$. In more detail, it holds that
$$
\frac{\beta^2}{2} + \sum_{j=2}^{i-1} \(d_j \cdot\(\beta + \sum_{k=2}^{j-1} d_k + \frac{d_j}{2}  \)  \) = \frac{1}{8} ~~ ,
$$
and
$$
\frac{\alpha^2}{2} + \sum_{j=i}^{n} \(d_j \cdot\(\alpha + \sum_{k=j+1}^{n} d_k + \frac{d_j}{2}  \)  \) = \frac{1}{8} ~~ .
$$
We conclude that the expression on the right hand side of Equation~\eqref{eq:uniform2} is exactly 1/4.
\end{proof}

\begin{figure}
\centering
\subfigure[Truthful location profile $\bx$.]
{
\label{sub:a_1}

\begin{tikzpicture}[scale=1.5]

\tikzstyle{blackdot}=[circle,draw=black,fill=black,thin,
inner sep=0pt,minimum size=1.5mm]
\tikzstyle{whitedot}=[circle,draw=black,fill=white,thin,
inner sep=0pt,minimum size=1.5mm]

\draw (0,0) circle (1cm);

\node (x1) at (-1,0) [blackdot] {};

\node (x2) at (330:1cm) [blackdot] {};
\node at (100:1cm) [blackdot] {};
\draw (100:1.3cm) node {\small{$x_i$}};

\node at (220:1cm) [blackdot] {};
\node at (205:1cm) [blackdot] {};
\node at (265:1cm) [blackdot] {};

\draw [color=white] (255:1.3cm) node {\small{$y$}};
\draw [color=white] (-1.3,0) node {\small{$l$}};
\draw [color=white] (1.3,0) node {\small{$\hat{l}$}};

\end{tikzpicture}
}
\hspace{1cm}
\subfigure[Manipulated location profile $\bx'$.]
{

\label{sub:b_1}

\begin{tikzpicture}[scale=1.5]

\tikzstyle{blackdot}=[circle,draw=black,fill=black,thin,
inner sep=0pt,minimum size=1.5mm]
\tikzstyle{whitedot}=[circle,draw=black,fill=white,thin,
inner sep=0pt,minimum size=1.5mm]
\tikzstyle{blackdot}=[circle,draw=black,fill=black,thin,
inner sep=0pt,minimum size=1.5mm]
\tikzstyle{graydot}=[circle,draw=gray,fill=gray,thin,
inner sep=0pt,minimum size=1.5mm]

\draw (0,0) circle (1cm);

\node (x1) at (-1,0) [blackdot] {};
\draw +(-1.3,0) node {\small{$l$}};
\node (hx1) at (1,0) [whitedot] {};
\draw +(1.3,0) node {\small{$\hat{l}$}};
\draw[dashed] (x1) -- (hx1);

\node (x2) at (330:1cm) [blackdot] {};
\draw (330:1.3cm) node {\small{$r$}};
\node (hx2) at (150:1cm) [whitedot] {};
\draw (150:1.3cm) node {\small{$\hat{r}$}};
\draw[dashed] (x2) -- (hx2);

\draw (255:0.9cm)--(255:1.1cm); 
\draw (255:1.3cm) node {\small{$y$}};

\node at (100:1cm) [graydot] {};
\draw (100:1.3cm) node {\small{$x_i$}};

\node at (220:1cm) [blackdot] {};
\node at (205:1cm) [blackdot] {};
\node at (310:1cm) [blackdot] {};
\draw (310:1.3cm) node {\small{$x_i'$}};
\node at (265:1cm) [blackdot] {};

\draw[color=white] (0:1.5) node {a};

%
%
%

\end{tikzpicture}

}
\caption{Illustration of the proof of Lemma~\ref{lem:notsemi_to_semi}.}
\label{fig:notsemi_to_semi}
\end{figure}

\begin{lemma}[Not semicircle to semicircle]
\label{lem:notsemi_to_semi}
Assume that $\bx\in G^n$ is such that the agents are not on one semicircle, but agent $i$ deviates such that in the new location profile $\bx'$ the agents are on one semicircle. Then
$$
\cost(\rc(\bx),x_i)\leq \cost(\lrm(\bx'),x_i) ~~ .
$$
\end{lemma}

\begin{proof}
By Lemma~\ref{lem:quarter} we have that $\cost(\rc(\bx),x_i)\leq 1/4$, therefore it will be sufficient to prove that $\cost(\lrm(\bx'),x_i)\geq 1/4$.

Let $l$ (for ``left'') and $r$ for ``right'' be the two extreme agent locations in $\bx'$, where $l\succeq r$. Note that since the agents were not on one semicircle under $\bx$, we have that $x_i\in \arc{\hat{l},\hat{r}}$. Let $y=\cen(l,r)$ be the center of the arc $\arc{l,r}$ (see Figure~\ref{fig:notsemi_to_semi}). We claim that $d(x_i,y)\geq 1/4$. This follows immediately from the facts that $d(\hat{l},y)\geq 1/4$, $d(\hat{r},y)\geq 1/4$, and $x_i\in \arc{\hat{l},\hat{r}}$. Hence, the cost of the mechanism is at least
$$
\cost(\lrm(\bx'),x_i)=\frac{1}{4}\cdot d(x_i,l) + \frac{1}{4} \cdot d(x_i,r) + \frac{1}{2} \cdot d(x_i,y) \geq \frac{1}{4}
\cdot(d(x_i,l)+d(x_i,r)) +  \frac{1}{2} \cdot \frac{1}{4} \geq \frac{1}{4} ~~ ,
$$
where the last transition follows from the fact that $d(x_i,l)+d(x_i,r)$ is the length of the \emph{long} arc between $l$ and $r$, therefore the value of this sum is at least 1/2.
\end{proof}

In order to deal with the last two deviations, we require one additional fundamental lemma. The lemma asserts that the cost of the RC Mechanism with respect to a point $y$ can only decrease if the point $y$ is added to the vector of locations. This is, in fact, the mirror image of Equation~\eqref{eq:antipod}, which is itself a special case of Lemma~\ref{lem:partition}.

\begin{lemma}
\label{lem:anti_wlog}
Let $\bx\in G^n$ such that $\bx$ is not on one semicircle, and let $y\in G$. Then
$$
\cost(\rc(\bx),y) \geq  \cost(\rc(\bx \cup \{y\}),y) ~~ .
$$
\end{lemma}

\begin{proof}
The cost incurred from the mechanism $\rc(\bx)$ with respect to $y$ is identical to the cost incurred from $\rc(\bx \cup \{y\})$ for all the intervals, except for the interval the point $\hy$ is on. Let $p$ and $q$ denote the antipodal points adjacent to $\hy$ such that $p \preceq \hy \preceq q$. Denote $\delta = d(p,\hy)$ and $\lambda = d(\hy,q)$, and assume without loss of generality that $\delta \leq \lambda$ (see Figure~\ref{fig:anti_wlog}). 

\begin{figure}
\begin{center}
\begin{tikzpicture}[scale=1.5]

\tikzstyle{blackdot}=[circle,draw=black,fill=black,thin,
inner sep=0pt,minimum size=1.5mm]
\tikzstyle{whitedot}=[circle,draw=black,fill=white,thin,
inner sep=0pt,minimum size=1.5mm]
\tikzstyle{blackdot}=[circle,draw=black,fill=black,thin,
inner sep=0pt,minimum size=1.5mm]
\tikzstyle{graydot}=[circle,draw=gray,fill=gray,thin,
inner sep=0pt,minimum size=1.5mm]

\draw (0,0) circle (1cm);

\draw[dashed] (300:1) -- (120:1);

\node at (300:1) [blackdot] {};
\draw (300:1.3) node {\small{$y$}};

\node at (120:1) [whitedot] {};
\draw (120:1.3) node {\small{$\hy$}};

\node at (70:1) [whitedot] {};
\draw (70:1.3) node {\small{$q$}};

\node at (150:1) [whitedot] {};
\draw (150:1.3) node {\small{$p$}};

\draw (95:0.85) node {\small{$\lambda$}};
\draw (135:0.85) node {\small{$\delta$}};

\node at (190:1) [whitedot] {};
\node at (240:1) [whitedot] {};
\node at (260:1) [whitedot] {};
\node at (45:1) [whitedot] {};
\node at (0:1) [whitedot] {};

\end{tikzpicture}
\end{center}
\caption{Illustration of the proof of Lemma~\ref{lem:anti_wlog}.}
\label{fig:anti_wlog}
\end{figure}

It is sufficient to show that the cost incurred when the random chosen point is on the arc $\carc{p,q}$ is lower under $\rc(\bx \cup \{y\})$ than under $\rc(\bx)$. The cost incurred by points on the arc $\carc{p,q}$ under $\rc(\bx \cup \{y\})$ is $\delta \left( d(y,p) + \delta/2 \right) + \lambda \left( d(y,q) + \lambda/2 \right)$.
The cost incurred by points on the arc $\carc{p,q}$ under $\rc(\bx)$ is $(\delta + \lambda) \left( d(y,q) + \halfx{\delta + \lambda} \right)$.
It holds that
\[
(\delta + \lambda) \left( d(y,q) + \halfx{\delta + \lambda} \right) \geq \delta \left( d(y,p) + \halfx{\delta} \right) + \lambda \left( d(y,q) + \halfx{\lambda} \right) \Leftrightarrow d(y,q) + \lambda \geq d(y,p).
\]
The inequality on the right hand side holds since $d(y,q) + \lambda = \half$ while $d(y,p) \leq \half$.
\end{proof}

\begin{lemma}[Semicircle to not semicircle]
\label{lem:semi_to_not_semi}
Assume that $\bx\in G^n$ is such that the agents are on one semicircle, but agent $i$ deviates such that in the new location profile $\bx'$ the agents are not on one semicircle. Then
$$
\cost(\lrm(\bx),x_i)\leq \cost(\rc(\bx'),x_i) ~~ .
$$
\end{lemma}

\begin{proof}
Let $x_1\succeq x_2\cdots \succeq x_n$ be the location of the $n$ agents, and denote $l=x_1$ (the leftmost agent) and $r=x_n$. Let $i\in N$, and assume without loss of generality that $d(x_i,r)\leq d(x_i,l)$, or equivalently $d(x_i,\h{l})\leq d(x_i,\h{r})$. Let $\alpha = d(r, \h{l})$, and $\beta=d(r,x_i)$. 
Finally, let $x_i'$ denote the new location of agent $i$ (see Figure~\ref{fig:semi_to_notsemi}).

We first calculate the cost of $x_i$ in the location profile $\bx$ under the LRM mechanism. We have that $d(x_i,l)=\half - \alpha - \beta$, $d(r,l)=\half - \alpha$, and 
$$
d(x_i,\cen(l,r))=\quartx{1} - \halfx{\alpha}-\beta ~~ .
$$ 
Then 
\begin{equation}
\cost(\lrm(\bx),x_i) = \quartx{\beta} + \quartx{\half - \alpha - \beta} + \halfx{\quartx{1} - \halfx{\alpha}-\beta} = \quartx{1} - \halfx{\alpha} - \halfx{\beta} ~~ .
\end{equation}

\begin{figure}
\begin{center}
\begin{tikzpicture}[scale=1.5]

\tikzstyle{blackdot}=[circle,draw=black,fill=black,thin,
inner sep=0pt,minimum size=1.5mm]
\tikzstyle{whitedot}=[circle,draw=black,fill=white,thin,
inner sep=0pt,minimum size=1.5mm]
\tikzstyle{blackdot}=[circle,draw=black,fill=black,thin,
inner sep=0pt,minimum size=1.5mm]
\tikzstyle{graydot}=[circle,draw=gray,fill=gray,thin,
inner sep=0pt,minimum size=1.5mm]

\draw (0,0) circle (1cm);

\node (x1) at (-1,0) [blackdot] {};
\draw +(-1.3,0) node {\small{$l$}};
\node (hx1) at (1,0) [whitedot] {};
\draw +(1.3,0) node {\small{$\hat{l}$}};
\draw[dashed] (x1) -- (hx1);

\node (x2) at (330:1cm) [blackdot] {};
\draw (330:1.3cm) node {\small{$r$}};
\node (hx2) at (150:1cm) [whitedot] {};
\draw (150:1.3cm) node {\small{$\hat{r}$}};
\draw[dashed] (x2) -- (hx2); 

\draw [dashed] (60:1) -- (240:1);
\node at (60:1cm) [blackdot] {};
\draw (60:1.3cm) node {\small{$x_i^*$}};

\draw [dashed] (300:1) -- (120:1);

\node at (300:1cm) [graydot] {};
\draw (300:1.3cm) node {\small{$x_i$}};

\node at (120:1) [whitedot] {};
\draw (120:1.3) node {\small{$\hx_i$}};

\node at (240:1) [whitedot] {};
\draw (240:1.3cm) node {\small{$\hx_i^*$}};

\draw (345:0.85cm) node {\small{$\alpha$}};
\draw (315:0.85cm) node {\small{$\beta$}};

\draw (270:0.85) node {\small{$\alpha+\beta$}};

\end{tikzpicture}
\end{center}
\caption{Illustration of the proof of Lemma~\ref{lem:semi_to_not_semi}.}
\label{fig:semi_to_notsemi}
\end{figure}
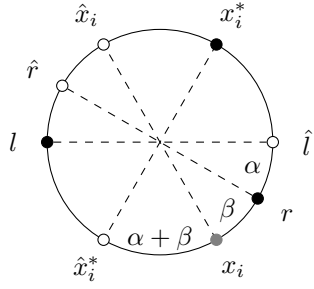

We now wish to give a lower bound on $\cost(\rc(\bx'),x_i)$. First, by Lemma~\ref{lem:anti_wlog} we have that 
$$
\cost(\rc(\bx'),x_i) \geq \cost(\rc(\bx'\cup\{x_i\}),x_i) ~~ .
$$
A subtle remark is that the above argument allows us to take $\h{r}$ into account even if $x_i=r$ and agent $i$ deviated to $x_i'$, a fact that allows us to avoid distinguishing this extreme case. 

Next, we would like to fix an ``optimal'' location for $\hx_i'$. First notice that since $\bx'$ is not on one semicircle, it must hold that $x_i'\in \arc{\h{r},\h{l}}$, that is, $\hx_i'\in\arc{r,l}$. Now, we ask: what is the location of $\hx_i'$ that yields the lowest cost with respect to $x_i$, when the RC Mechanism chooses a point in $G\setminus \arc{\h{l},\h{r}}$ (the long arc between $\h{l}$ and $\h{r}$)? There are no antipodal points, other than $\hx_i'$, in $G\setminus \arc{\h{l},\h{r}}$. Furthermore, we can treat this arc as an interval with respect to distances from $x_i$, since the short arcs $\carc{x_i,\h{l}}$ and $\carc{x_i,\h{r}}$ are contained in $G\setminus \arc{\h{l},\h{r}}$. Therefore, we can apply Lemma~\ref{lem:partition} (by normalizing the length of $G\setminus \arc{\h{l},\h{r}}$). In particular, let $\hx_i^*\succeq x_i$ such that 
$$
d(\hx_i^*,x_i)=d(x_i,\h{l})=\alpha+\beta ~~ ,
$$
and let $\bx^*=\langle x_i^*,x_{-i}\rangle$; see Figure~\ref{fig:semi_to_notsemi} for an illustration. From (the proof of) Lemma~\ref{lem:partition} it follows that 
$$
\cost(\rc(\bx'\cup\{x_i\}),x_i) \geq \cost(\rc(\bx^*\cup\{x_i\}),x_i) ~~ .
$$
Hence, it is sufficient to prove that 
\begin{equation}
\label{eq:lower_bound_xopt}
\cost(\rc(\xnewopt\cup\{x_i\}),x_i) \geq \quartx{1} - \halfx{\alpha} - \halfx{\beta} ~~ .
\end{equation}

We calculate the cost of the mechanism $\rc(\xnewopt \cup \{x_i\})$ with respect to agent $i$. Let $y$ be the point chosen by the mechanism. By the choice of $\xopt$, if $y$ is on the arc $\carc{\h{l},\hx_i^*}$, the cost is zero. If $y \in \carc{\h{l},\hx_i}$, we apply Lemma~\ref{lem:interval} with $y_1=x_i$, $y_{m+1}=\hx_i$, $y_j = \hx_{j-1}$ for $j=2, \ldots, m+1$, and $d_1=\alpha+\beta$. We get:
\[
\halfx{(\alpha+\beta)^2} + \sum_{j=2}^{m} \(d_j \cdot\(\alpha + \beta + \sum_{k=2}^{j-1} d_k + \frac{d_j}{2}  \)  \) = \frac{1}{8} ~~ .
\]
Thus,
\begin{equation}
\sum_{j=2}^{m} \(d_j \cdot\(\alpha + \beta + \sum_{k=2}^{j-1} d_k + \frac{d_j}{2}  \)  \) = \frac{1}{8} - \halfx{(\alpha+\beta)^2} ~~ .
\label{eq:cost_semi}
\end{equation}
Similarly, if $y \in \carc{\hx_i,\hx^*_i}$, we apply Lemma~\ref{lem:interval} with $y_1=x_i$, $y_2=\hx^*_i$, $y_{m+1}=\hx_i$, $y_j = \hx_{n-j+3}$ for $j=3,\ldots,m+1$, and $d_1=\alpha+\beta$. We get the same expected cost as in Equation~\eqref{eq:cost_semi}.

Taken together, 
$$
\cost(\rc(\xnewopt \cup \{x_i\}),x_i) = 2 \left(\frac{1}{8} - \halfx{(\alpha+\beta)^2}\right) = \quartx{1} - (\alpha + \beta)^2 ~~ .
$$ 
It holds that
\[
\quartx{1} - (\alpha + \beta)^2 \geq \quartx{1} - \halfx{\alpha} - \halfx{\beta} \Leftrightarrow \alpha + \beta \leq \half,
\]
but the last inequality follows directly from the fact that $x_i\in\carc{r,l}$. This establishes Equation~\eqref{eq:lower_bound_xopt}, and thus completes the proof of Lemma~\ref{lem:semi_to_not_semi}.
\end{proof}

\begin{lemma}[Not semicircle to not semicircle]
\label{lem:not_semi_to_not_semi}
Assume that $\bx\in G^n$ is such that the agents are not on one semicircle, and agent $i$ deviates such that in the new location profile $\bx'$ the agents are not on one semicircle. Then
$$
\cost(\rc(\bx),x_i)\leq \cost(\rc(\bx'),x_i) ~~ .
$$
\end{lemma}

\begin{proof}
The proof of the lemma follows quite directly from the previous lemmata. Indeed, by Lemma~\ref{lem:anti_wlog} we have that 
$$
\cost(\rc(\bx'\cup\{x_i\}),x_i)\leq \cost(\rc(\bx'),x_i) ~~ .
$$
Let $p,q\in G$ be the two antipodal points adjacent to $x_i$. By Lemma~\ref{lem:interval}, for any $\hx_i'\in G\setminus \carc{p,q}$, 
$$
\cost(\rc(\bx),x_i)=\cost(\rc(\bx'\cup\{x_i\}),x_i) ~~ .
$$
Hence, it sufficient to handle the case where $\hx_i'\in\carc{p,q}$. Notice that $x_i\in\carc{p,q}$, that is, $x_i$ is in the \emph{short} arc between $p$ and $q$, by our assumption that the agents in $\bx$ are not on one semicircle. Therefore, we can apply Lemma~\ref{lem:partition} to this arc (by normalizing its length, and replacing $x$ in the lemma by $x_i$ and $y$ by $\hx_i'$). It follows that the optimal location for $\hx_i'$ is the edges of the arc $\carc{p,q}$, which directly means that 
$$
\cost(\rc(\bx),x_i)\leq\cost(\rc(\bx'\cup\{x_i\}),x_i) ~~ . 
$$
\end{proof}

\section{Proof of Theorem~\ref{thm:mm_rand_lb}}
\label{app:mm_rand_lb}

Let $m,k\in\mathbb{N}$, whose value remains to be determined. The construction of the graph $G$ is recursive, and depends on $m$ and $k$. We start with an edge of length 1, which connects the vertices $l^0$ (for ``left'') and $r^0$ (for ``right''). The vertex $l^0$ is connected to $m$ vertices via edges of length 1; these vertices are called \emph{left vertices on level 1}. Each left vertex on level 1 is connected to $m$ vertices via $m$ edges of length 2; these vertices are called \emph{left vertices on level 2}. In general, each left vertex on level $d$ is connected to $m$ left vertices on level $d+1$ via edges of length $2^d$. The maximum level is $k$, that is, the left vertices on level $k$ are leaves. The construction is symmetric with respect to the right vertices, i.e., $r^0$ is connected to $m$ right vertices on level 1 via edges of length 1, and so on.

Now, let $f:G^n\rightarrow \Delta(G)$ be a randomized SP mechanism. Assume for ease of exposition that $2m^k$ divides $n$, and consider a location profile $\bx^0$ where there are $n/2$ agents at $l^0$ and $n/2$ agents at $r^0$; this profile is illustrated in Figure~\ref{sub:a}. Clearly, since the distance between $l^0$ and $r^0$ is 1, we have that $\exp[d(f(\bx^0),l^0)]\geq 1/2$, or $\exp[d(f(\bx^0),r^0)]\geq 1/2$; assume without loss of generality that the former statement is true.

\begin{figure}[t]
\centering
\subfigure[Location profile $\bx^0$.]
{
\label{sub:a}

\begin{tikzpicture}[scale=1]
\tikzstyle{dot}=[circle,draw=black,fill=white,thin,
inner sep=0pt,minimum size=5mm]

\node (l0) at (0,0) [dot] {\tiny{$\frac{n}{2}$}};
\node at (0.2,-0.4) {\small{$l^0$}};
\node (r0) at (1,0) [dot] {\tiny{$\frac{n}{2}$}};
\node at (0.8,-0.4) {\small{$r^0$}};
\path (l0) edge node [above] {\tiny{1}} (r0);

\node (l1) at (-0.7,0.7) [dot] {};

\node (l11) at (0,2) [dot] {};
\node (l12) at (-0.7,2.4) [dot] {};
\node (l13) at (-1.4,2) [dot] {};

\path (l1) edge node [right] {\tiny{2}} (l11);
\path (l1) edge node [right] {\tiny{2}} (l12);
\path (l1) edge node [right] {\tiny{2}} (l13);

\node (l2) at (-1,0) [dot] {};

\node (l21) at (-2.3,0.7) [dot] {};
\node (l22) at (-2.7,0) [dot] {};
\node (l23) at (-2.3,-0.7) [dot] {};

\path (l2) edge node [above] {\tiny{2}} (l21);
\path (l2) edge node [above] {\tiny{2}} (l22);
\path (l2) edge node [above] {\tiny{2}} (l23);

\node (l3) at (-0.7,-0.7) [dot] {};

\node (l31) at (0,-2) [dot] {};
\node (l32) at (-0.7,-2.4) [dot] {};
\node (l33) at (-1.4,-2) [dot] {};

\path (l3) edge node [right] {\tiny{2}} (l31);
\path (l3) edge node [right] {\tiny{2}} (l32);
\path (l3) edge node [right] {\tiny{2}} (l33);

\path (l0) edge node [above] {\tiny{1}} (l1);
\path (l0) edge node [above] {\tiny{1}} (l2);
\path (l0) edge node [above] {\tiny{1}} (l3);

\node (r1) at (1.7,0.7) [dot] {};

\node (r11) at (2.4,2) [dot] {};
\node (r12) at (1.7,2.4) [dot] {};
\node (r13) at (1,2) [dot] {};

\path (r1) edge node [left] {\tiny{2}} (r11);
\path (r1) edge node [left] {\tiny{2}} (r12);
\path (r1) edge node [left] {\tiny{2}} (r13);

\node (r2) at (2,0) [dot] {};

\node (r21) at (3.3,0.7) [dot] {};
\node (r22) at (3.7,0) [dot] {};
\node (r23) at (3.3,-0.7) [dot] {};

\path (r2) edge node [above] {\tiny{2}} (r21);
\path (r2) edge node [above] {\tiny{2}} (r22);
\path (r2) edge node [above] {\tiny{2}} (r23);

\node (r3) at (1.7,-0.7) [dot] {};

\node (r31) at (2.4,-2) [dot] {};
\node (r32) at (1.7,-2.4) [dot] {};
\node (r33) at (1,-2) [dot] {};

\path (r3) edge node [left] {\tiny{2}} (r31);
\path (r3) edge node [left] {\tiny{2}} (r32);
\path (r3) edge node [left] {\tiny{2}} (r33);

\path (r0) edge node [above] {\tiny{1}} (r1);
\path (r0) edge node [above] {\tiny{1}} (r2);
\path (r0) edge node [above] {\tiny{1}} (r3);

\end{tikzpicture}
}
\hspace{1cm}
\subfigure[Location profile $\bx^1$.]
{

\label{sub:b}

\begin{tikzpicture}[scale=1]
\tikzstyle{dot}=[circle,draw=black,fill=white,thin,
inner sep=0pt,minimum size=5mm]

\node (l0) at (0,0) [dot] {};
\node at (0.2,-0.4) {\small{$l^0$}};
\node (r0) at (1,0) [dot] {\tiny{$\frac{n}{2}$}};
\node at (0.8,-0.4) {\small{$r^0$}};
\path (l0) edge node [above] {\tiny{1}} (r0);

\node (l1) at (-0.7,0.7) [dot] {\tiny{$\frac{n}{2m}$}};

\node (l11) at (0,2) [dot] {};
\node (l12) at (-0.7,2.4) [dot] {};
\node (l13) at (-1.4,2) [dot] {};

\path (l1) edge node [right] {\tiny{2}} (l11);
\path (l1) edge node [right] {\tiny{2}} (l12);
\path (l1) edge node [right] {\tiny{2}} (l13);

\node (l2) at (-1,0) [dot] {\tiny{$\frac{n}{2m}$}};

\node (l21) at (-2.3,0.7) [dot] {};
\node (l22) at (-2.7,0) [dot] {};
\node (l23) at (-2.3,-0.7) [dot] {};

\path (l2) edge node [above] {\tiny{2}} (l21);
\path (l2) edge node [above] {\tiny{2}} (l22);
\path (l2) edge node [above] {\tiny{2}} (l23);

\node (l3) at (-0.7,-0.7) [dot] {\tiny{$\frac{n}{2m}$}};

\node (l31) at (0,-2) [dot] {};
\node (l32) at (-0.7,-2.4) [dot] {};
\node (l33) at (-1.4,-2) [dot] {};

\path (l3) edge node [right] {\tiny{2}} (l31);
\path (l3) edge node [right] {\tiny{2}} (l32);
\path (l3) edge node [right] {\tiny{2}} (l33);

\path (l0) edge node [above] {\tiny{1}} (l1);
\path (l0) edge node [above] {\tiny{1}} (l2);
\path (l0) edge node [above] {\tiny{1}} (l3);

\node (r1) at (1.7,0.7) [dot] {};

\node (r11) at (2.4,2) [dot] {};
\node (r12) at (1.7,2.4) [dot] {};
\node (r13) at (1,2) [dot] {};

\path (r1) edge node [left] {\tiny{2}} (r11);
\path (r1) edge node [left] {\tiny{2}} (r12);
\path (r1) edge node [left] {\tiny{2}} (r13);

\node (r2) at (2,0) [dot] {};

\node (r21) at (3.3,0.7) [dot] {};
\node (r22) at (3.7,0) [dot] {};
\node (r23) at (3.3,-0.7) [dot] {};

\path (r2) edge node [above] {\tiny{2}} (r21);
\path (r2) edge node [above] {\tiny{2}} (r22);
\path (r2) edge node [above] {\tiny{2}} (r23);

\node (r3) at (1.7,-0.7) [dot] {};

\node (r31) at (2.4,-2) [dot] {};
\node (r32) at (1.7,-2.4) [dot] {};
\node (r33) at (1,-2) [dot] {};

\path (r3) edge node [left] {\tiny{2}} (r31);
\path (r3) edge node [left] {\tiny{2}} (r32);
\path (r3) edge node [left] {\tiny{2}} (r33);

\path (r0) edge node [above] {\tiny{1}} (r1);
\path (r0) edge node [above] {\tiny{1}} (r2);
\path (r0) edge node [above] {\tiny{1}} (r3);

\end{tikzpicture}

}
\caption{An Illustration of the proof of Theorem~\ref{thm:mm_rand_lb}, for $m=3$ and $k=2$. A number inside a node indicates the number of agents that are located at this node.}
\label{fig:mm_rand_lb}
\end{figure}
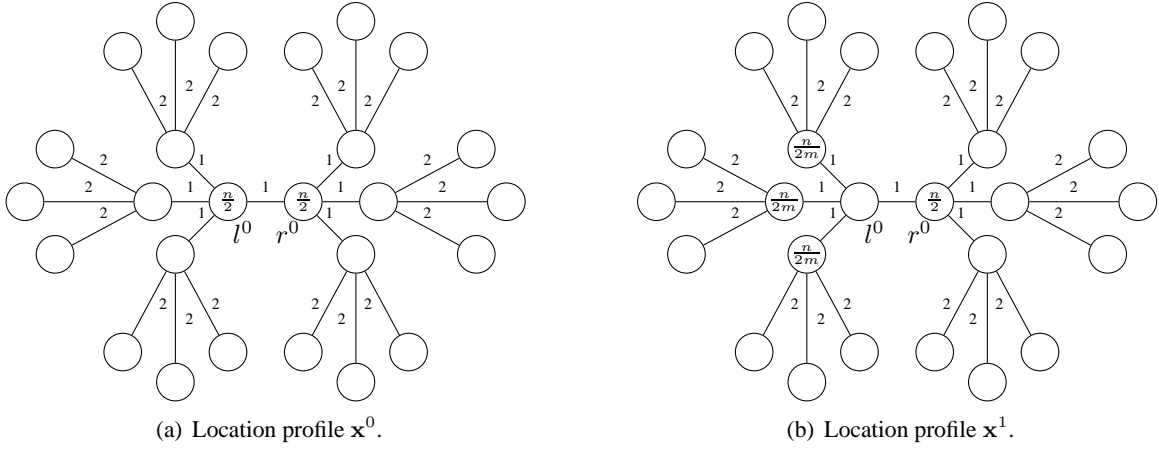

Next, consider the location profile $\bx^1$, where we have $n/2m$ agents in each left vertex on level 1, and $n/2$ agents in $r^0$; this profile is illustrated is Figure~\ref{sub:b}. We claim that it still holds that $\exp[d(f(\bx^0),l^0)]\geq 1/2$. Indeed, $\bx^1$ can be obtained from $\bx^0$ by moving the agents, one by one, from $l^0$ to the left vertices on level 1. Due to strategyproofness, the expected distance from $l^0$ cannot decrease after each deviation. Since the expected distance was initially at least 1/2, this is true after the $n/2$ agents have all deviated.

Denote the left vertices on level 1 by $l^1_1,\ldots,l^1_m$. For any point $y\in G$ such that $d(y,l^0)\geq 1/2$, we have that
$$
\sum_{j=1}^m d(y,l^1_j) \geq m + (m-2) \cdot \frac{1}{2} = \frac{3m-2}{2} ~~ .
$$
Since this is true point-wise, it is also true that the sum of expected distances between $f(\bx^1)$ and the left vertices on level 1 is at least $\frac{3m-2}{2}$, that is,
$$
\sum_{j=1}^m \exp[d(f(\bx^1),l^1_j)] \geq \frac{3m-2}{2} ~~ .
$$
By averaging over these $m$ vertices we conclude that there exists a left vertex on level 1, denoted $l^1$, such that $\exp[d(f(\bx^1),l^1)]\geq 3/2-1/m$.

We subsequently consider the location profile $\bx^2$ that is obtained from $\bx^1$ by moving the $n/(2m)$ agents from $l^1$ to its $m$ neighbors on level 2, such that each left vertex on level 2 that is adjacent to $l^1$ has $n/(2m^2)$ agents. By similar arguments as before, we get that $\exp[d(f(\bx^2),l^1)]\geq 3/2-1/m$, and, therefore, that there exists a left vertex on level 2 that is adjacent to $l^1$, call it $l^2$, such that
$$
\exp[d(f(\bx^2),l^2)]\geq \frac{2m + (m-2)\left(\frac{3}{2}-\frac{1}{m}\right)}{m}\geq \frac{7}{2} - \frac{4}{m} ~~ .
$$

We inductively build location profiles $\bx^3,\bx^4,\ldots,\bx^k$ in this fashion. We then have the following claim.

\begin{lemma}
\label{lem:ind}
There exists a left vertex on level $k$, denoted $l^k$, such that
\begin{equation}
\label{eq:k}
\exp[d(f(\bx^k),l^k)]\geq \frac{2^{k+1}-1}{2} - \frac{2^{k+1}-(k+2)}{m} ~~ .
\end{equation}
\end{lemma}

\begin{proof}
We prove the lemma by induction on the level $d$. For $d=0$, we have that
$$
\frac{2^{0+1}-1}{2} -\frac{2^{0+1}-(0+2)}{m} = \frac{1}{2} ~~ ,
$$
which is indeed the lower bound that we have obtained for $\exp[d(f(\bx^0),l^0)]$.

Let us assume that there exists a left vertex on level $d$, denoted $l^d$, such that
$$
\exp[d(f(\bx^d),l^d)]\geq \frac{2^{d+1}-1}{2} - \frac{2^{d+1}-(d+2)}{m} ~~ .
$$
Constructing $\bx^{d+1}$ along the lines given above, we get that there is a left vertex on level $d+1$, denoted $l^{d+1}$, such that
\begin{align*}
\exp[d(f(\bx^{d+1}),l^{d+1})]& \geq \frac{1}{m}\left[2^d\cdot m + (m-2)\left( \frac{2^{d+1}-1}{2} - \frac{2^{d+1}-(d+2)}{m} \right) \right]\\
& \geq \frac{2^{d+2}-1}{2} - \frac{2^{d+2}-(d+3)}{m} ~~ . \qedhere
\end{align*}
\end{proof}

From Lemma~\ref{lem:ind} we obtain a location profile $\bx^k$ and a left vertex on level $k$, $l^k$, such that Equation~\eqref{eq:k} holds, hence the expected maximum distance is at least the expression at the right hand side of Equation~\eqref{eq:k}. On the other hand, under $\bx^k$ the solution that locates the facility at $l^{k-1}$ has a maximum cost of $2^{k-1}$. The ratio is at least
$$
\frac{\mc(f(\bx^k),\bx^k)}{\mc(l^{k-1},\bx^k)}\geq\frac{\frac{2^{k+1}-1}{2} - \frac{2^{k+1}-(k+2)}{m}}{2^{k-1}} \geq 2 - \frac{1}{2^k} - \frac{4}{m} ~~ .
$$
Note that we have to choose $m$ and $k$ such that $n/(2m^k)\geq 1$, so that we still have at least one agent at each of $m$ left vertices on level $k$ in our construction. By taking $k=\Theta\left(\sqrt{\log n}\right)$, and
$$
m = \Theta\left(n^{1/k}\right) = \Theta\left(n^{\frac{1}{\sqrt{\log n}}}\right) ~~ ,
$$
we satisfy the above constraint and get that the approximation ratio of $f$ is as announced.
\qed

\begin{thebibliography}{10}

\bibitem{ARR98}
S.~Arora, P.~Raghavan, and S.~Rao.
\newblock Approximation schemes for {E}uclidean $k$-medians and related
  problems.
\newblock In {\em Proceedings of the 30th Annual ACM Symposium on the Theory of
  Computing (STOC)}, pages 106--113, 1998.

\bibitem{BE96}
M.~Bern and D.~Eppstein.
\newblock Approximation algorithms for geometric problems.
\newblock In D.~Hochbaum, editor, {\em Approximation Algorithms for {NP}-Hard
  Problems}. PWS Publishing, 1996.

\bibitem{BL86}
H.~J. Brandelt and M.~Labb\'e.
\newblock How bad can a voting location be.
\newblock {\em Social Choice and Welfare}, 3:125--145, 1986.

\bibitem{DFP08}
O.~Dekel, F.~Fischer, and A.~D. Procaccia.
\newblock Incentive compatible regression learning.
\newblock In {\em Proceedings of the 19th Annual ACM-SIAM Symposium on Discrete
  Algorithms (SODA)}, pages 277--286, 2008.

\bibitem{Gib73}
A.~Gibbard.
\newblock Manipulation of voting schemes.
\newblock {\em Econometrica}, 41:587--602, 1973.

\bibitem{HT81}
P.~Hansen and J.-F. Thisse.
\newblock Outcomes of voting and planning: {C}ondorcet, {W}eber and {R}awls
  locations.
\newblock {\em Journal of Public Economics}, 16:1--15, 1981.

\bibitem{Lab85}
M.~Labb\'e.
\newblock Outcomes of voting and planning in single facility location problems.
\newblock {\em European Journal of Operations Research}, 20:299--313, 1985.

\bibitem{Meir08}
R.~Meir.
\newblock Strategy proof classification.
\newblock Master's thesis, The Hebrew University of Jerusalem, 2008.
\newblock Available from:
  http://www.cs.huji.ac.il/\verb+~+reshef24/spc.thesis.pdf.

\bibitem{Nis07}
N.~Nisan.
\newblock Introduction to mechanism design (for computer scientists).
\newblock In N.~Nisan, T.~Roughgarden, \'E. Tardos, and V.~Vazirani, editors,
  {\em Algorithmic Game Theory}, chapter~9. Cambridge University Press, 2007.

\bibitem{PT09}
A.~D. Procaccia and M.~Tennenholtz.
\newblock Approximate mechanism design without money.
\newblock In {\em Proceedings of the 10th ACM Conference on Electronic Commerce
  (ACM-EC)}, 2009.
\newblock To appear.

\bibitem{Roth07}
M.~Rothkopf.
\newblock Thirteen reasons the {V}ickrey-{C}larke-{G}roves process is not
  practical.
\newblock {\em Operations Research}, 55(2):191--197, 2007.

\bibitem{Sat75}
M.~Satterthwaite.
\newblock Strategy-proofness and {A}rrow's conditions: Existence and
  correspondence theorems for voting procedures and social welfare functions.
\newblock {\em Journal of Economic Theory}, 10:187--217, 1975.

\bibitem{SV04}
J.~Schummer and R.~V. Vohra.
\newblock Strategy-proof location on a network.
\newblock {\em Journal of Economic Theory}, 104(2):405--428, 2004.

\bibitem{SV07}
J.~Schummer and R.~V. Vohra.
\newblock Mechanism design without money.
\newblock In N.~Nisan, T.~Roughgarden, \'E. Tardos, and V.~Vazirani, editors,
  {\em Algorithmic Game Theory}, chapter~10. Cambridge University Press, 2007.

\end{thebibliography}
\end{document}